\documentclass[11pt]{article}
\usepackage{jheppub}
\usepackage[dvipsnames]{xcolor}
\usepackage[T1]{fontenc} % if needed
\usepackage{mathrsfs} % for mathscr

\usepackage{mathtools}
\usepackage{nccmath}
\usepackage{amsmath}
\usepackage{amssymb}
\usepackage{amsthm}
\usepackage{graphicx}
\usepackage{subcaption}
\newtheorem{theorem}{Theorem}
\newtheorem{remark}{Remark}

\def\E{\mathrm{E}}

\def\L{\mathrm{L}}

\def\I{\mathcal{I}}

\title{\centering Landau Analysis in Momentum Space \\ with Massless Particles: an \emph{Amuse Bouche}}

\author{Cristian~Vergu${}^{1, 2}$}%
\emailAdd{c.vergu@gmail.com}

\affiliation{\small $^1$ Institute for Gravitation and the Cosmos, Department of Physics, \\ Pennsylvania State University,\\ University Park, Pennsylvania 16802, USA}
\affiliation{\small $^{2}$ Max-Planck-Institut f\"ur Physik, \\ Werner-Heisenberg-Institut,\\ Boltzmannstr.~8, 85748 Garching, Germany}

\abstract{We illustrate how methods from Landau analysis that have been developed for studying the properties of massive Feynman integrals in momentum space can be generalized to massless integrals.  We consider integrals with both massive and massless propagators in arbitrary dimensions, paying attention to square root branch points. By focusing on a number of well-chosen examples, we show how resolution of singularities (via blow-ups or complex structure deformation) can be used to predict how the behavior of these integrals is modified as different numbers of propagators are chosen to be massless.}

\begin{document}

\maketitle

\section{Introduction}
\label{sec:intro}

One class of Feynman integrals which is better understood than the rest is that whose elements can be written as iterated integrals of $d \log$ forms.  In most usual cases these are polylogarithmic functions (but see ref.~\cite{Duhr:2020gdd} for a counterexample).  For these functions one can use the notion of symbol (see ref.~\cite{Goncharov:2010jf}) to obtain a simpler and (almost) canonical\footnote{The qualifier ``almost'' here means the following.  In practice one can take the symbol entries to be elements of a ring of integers of a field.  If one knows how to do factorization in this ring effectively, then the problem of putting a symbol in a canonical form can be considered solved.  Alas, this is not always easy in practice.  See ref.~\cite{Bourjaily:2019igt}, where one had to compute factorization in terms of non-principal ideals in a number field which was an iterated quadratic extension of high order.  Similar difficulties arise in cluster algebras (see ref.~\cite{Golden:2013xva}) where the process of cluster mutations involves division which is always exact but not always easily computable.} form of the answer.  At any rate, elliptic or Calabi-Yau type integrals are much less understood and we will not discuss them in the following.

It had become clear early on, already with ref.~\cite{Gaiotto:2011dt}, that singularities of integrals are neatly encoded in the symbol.  In the case analyzed in that reference, the first entries of the remainder function were noticed to be ``distances'' of type $x_{i j}^2$ where $x$'s are dual space coordinates.  The study of more general Landau singularities was taken up later in a series of papers~\cite{Dennen:2015bet, Dennen:2016mdk, Prlina:2017azl, Prlina:2017tvx, Chicherin:2025cua}.  More constraints of Steinmann type (see refs.~\cite{Steinmann, Steinmann2}) were used in later papers.  In fact, Steinmann type constraints were used even before in ref.~\cite{Bartels:2009vkz} to argue that the BDS ansatz (ref.~\cite{Bern:2005iz}) needs to be corrected by a remainder function.

The notion of Steinmann compatibility and some heuristic extension of it were used very effectively to constrain ans\"atze of symbols on which other types of constraints were imposed to fully determine the symbol (and ultimately the function).  See refs.~\cite{Caron-Huot:2018dsv, Morales:2022csr, Dixon:2022rse} for a number of results in this direction.

However, the notion of Steinmann compatibility has a vast but little-known generalization which was derived in ref.~\cite{pham} (see also the discussion in ref.~\cite{Hannesdottir:2022xki} and refs.~\cite{pham2011singularities, pham1968singularities} for useful background for ref.~\cite{pham}).  These Steinmann-Pham compatibility conditions (at least for the generic mass case) should tell us which sequential discontinuities can yield non-vanishing answers.  Another constraint arises from the hierarchical principle of ref.~\cite{Landshoff1966} (see also ref.~\cite{pham} and ref.~\cite{Berghoff:2022mqu} for a more recent discussion).

Yet another constraint arises from the study of the behavior near a Landau singularity.  This study was already initiated in ref.~\cite{Landau:1959fi} and later more detailed derivations were done in refs.~\cite{PolkinghorneScreaton, ELOP} and, in parallel, in the mathematical literature in ref.~\cite{BSMF_1959__87__81_0,pham2011singularities}.  In ref.~\cite{Hannesdottir:2021kpd}, the Landau exponent, which controls the behavior near a singularity (and which has a simple expression in terms of the number of loops, the dimension and the number of propagators going on-shell), was shown to be tied to the location of a given singularity in the symbol.

This list of constraints makes one hopeful that a Landau bootstrap, as advocated in ref.~\cite{Hannesdottir:2024hke} should be possible.  Indeed, it is hard to imagine what kind of obstruction could prevent it from being successful for the class of integrals which can be written in an iterated integral form.  This would enable one to build amplitudes or form factors from \emph{first principles}, without recourse to ad-hoc tools like input from integrability (which is only available for planar $\mathcal{N} = 4$ theory).

Indeed, a detailed understanding of singularities is needed (see refs.~\cite{landshoff-olive, PhysRev.125.2139, pham}) in order to know how to do the analytic continuation (which way the iterated integral contour should avoid the singularities of the differential forms occurring in the symbol).

Despite their great promise, the constraints mentioned above have been proved in a restricted setting where no ``fine tuning'' of the parameters is allowed.  Still, such fine tuning does routinely occur in physical situations; in and out particles can have equal masses, etc.  It is less clear what can be said in general about these fine-tuned cases.

In this paper, as a step towards a general theory, we study Landau singularities of Feynman integrals, in various singular, non-generic, kinematic configurations.  One such non-generic kinematic configuration is when we have massless particles.  Integrals with massless particles yield in general simpler answers for the integrals than their massive counterparts, but they give rise to some difficulties, such as infrared divergences, which need to be regulated.

From the point of view of on-shell spaces, which have been a central notion in ref.~\cite{pham, Hannesdottir:2022xki}, one of the complications arising for massless particles is the possibility that the on-shell spaces are singular.  Indeed, the on-shell space of one massless propagator is a light-cone, which has a singularity at the tip of the light-cone.  When the on-shell spaces are singular, the general results of ref.~\cite{pham, pham2011singularities} do not immediately apply.

The best understood case for Landau singularities is called a \emph{simple pinch}, which is such that the $\alpha_e$ and the internal momenta $q_e$ taking part in a pinch are uniquely determined and a certain Hessian matrix is positive definite.  For non-generic kinematics several new possibilities arise.  For example, one may obtain a non-unique pinch, as in the case of the massless bubble.  Even when the critical point is unique, the Hessian matrix obtained by expansion around it may not be definite.  For example, it could have zero eigenvalues.

There are two main ways one can deal with singularities: blow-ups or deformations.  We describe both.  While we describe deformations in more detail, it is likely that blow-ups will lead to a more powerful algorithm for understanding singularities.  The blow-ups make all the on-shell spaces non-singular by replacing light-cones with cylinders where the tip of the light-cone becomes a sphere.  This however changes the integrand, by replacing massless on-shell propagators by products, as described in sec.~\ref{sec:landau-eq-blowups}.  One then has many ways to solve the Landau equations, by taking for each massless propagator a soft, collinear or soft-collinear condition.  There appear to be similarities to the method of regions of ref.~\cite{Beneke:1997zp} and with tropical geometry approaches of ref.~\cite{Arkani-Hamed:2022cqe}, but we will not seek to understand these connections here.

The second way to deal with singularities are (complex structure) deformations.  This can be done in several ways.  One way is to make the massless propagators massive and study the limit where the masses vanish.  Another way, involves doing an expansion in a small parameter which is the distance to the Landau loci.  This works because away from the Landau locus the intersections of on-shell spaces is generically smooth.  A phenomenon which can be conveniently analyzed by this method is what happens when the Hessian matrix around a critical point is itself singular at the Landau locus.  Since we are expanding around a point in the neighborhood of the Landau locus, we can compute the asymptotics of $\det H$ where $H$ is the Hessian.  We check in a few cases that this analysis gives the right answer.

We check the Pham-Steinmann relations in a few cases.  We also check the hierarchical principle, in the rather exotic case of tadpole singularities.  We show that these constraints are very stringent, fixing the final answer uniquely (once we accept the premise that the answer is of polylogarithmic type).  The hierarchical principle also plays a big role, but there is more to be understood when massless particles play an essential role.  We hope to return to this in the future.

In sec.~\ref{sec:second_type_inversion} we describe an approach to second type and mixed second type singularities via inversion, which reduces the analysis to a singularity at infinity to the study of a singularity at the origin.  Second type singularities naturally arise in dimensional regularization, even for integrands which are dual conformal invariant in integer dimensions and the associated singularities are very difficult to understand by other methods (see ref.~\cite{Caron-Huot:2014lda}).  This has also been explored in ref.~\cite{Hannesdottir:2024hke}.

Finally in sec.~\ref{sec:landau-geometric} we present an approach to solving Landau equations based on volumes of various simplices with fixed side lengths.  The simplices can be in Euclidean, Lorentzian or other signatures.  This approach is inspired by that presented in ref.~\cite{Okun:1960cls} in a two-dimensional example.

\paragraph{Notation}
We denote a Feynman integral corresponding to a graph $G$ with edges $e$, each of mass $m_e$ in $D$ spacetime dimensions with
\begin{equation}
    \I_G = \int_h \prod_{\ell=1}^{\L} d^D k_\ell \frac{N}{\prod_{e=1}^{\E} (q_e^2-m_e^2)}
    \label{eq:I_int}
\end{equation}
where the $k_\ell$, with $\ell \in \{1,2, \ldots, \L \}$, form a basis of loop momenta, and $q_e$ is the momentum flowing through edge $e$. The numerator $N$ is a function of the internal momenta $q_e$ and the external momenta $p_e$.  We have denoted by $h$ the contour of integration, which avoids the singularities in the denominators according to the Feynman $i \varepsilon$ prescription.

\section{The Landau Equations with Blow-ups}
\label{sec:landau-eq-blowups}

\subsection{Pinches and permanent pinches}

The Landau equations are easiest to understand for Feynman integrals for which all internal edges are massive. In those cases, we have a number of singular surfaces at the locations $q_e^2-m_e^2=0$, given by the vanishing loci of the denominators of~\eqref{eq:I_int}.

As a simple example, let us look at the first edge, whose corresponding singular surface given by $q_1^2-m^2=0$. Assuming for simplicity that $q_1$ is one of our basis loop momenta, we can draw this surface in the plane of $q_1^0$ and $q_1^1$ as in fig.~\ref{fig:eyelash-massive}.

\begin{figure}
  \centering
  \includegraphics{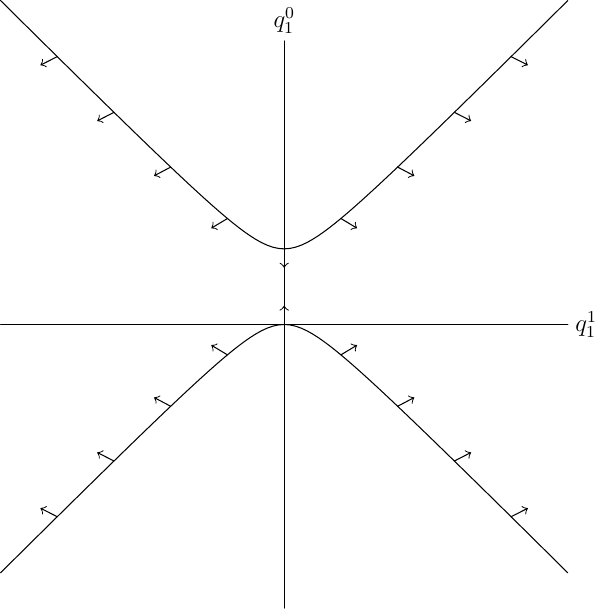}
  \caption{The on-shell surface of a massive particle, together with a choice of complex detour which avoids all the complex points of the singular on-shell hypersurface.}
  \label{fig:eyelash-massive}
\end{figure}

The arrows that we have drawn on the singular surface $q_1^2-m^2=0$ form a vector field that represents the imaginary part we have to add to the contour $h$ to have it avoid the singular surface consistent with the causal $i \varepsilon$ prescription.

A problem with this picture arises already if we take this propagator to be massless. The singular surfaces can then be drawn as in fig.~\ref{fig:eyelash-massless}.  In this case, we see that there is no consistent way of making a continuous deformation of the contour to complex values, which avoids the singularity when $q_1$ becomes soft ($q_1^i=0$ for $i = 0, 1, \ldots D-1$).

\begin{figure}
  \centering
  \includegraphics{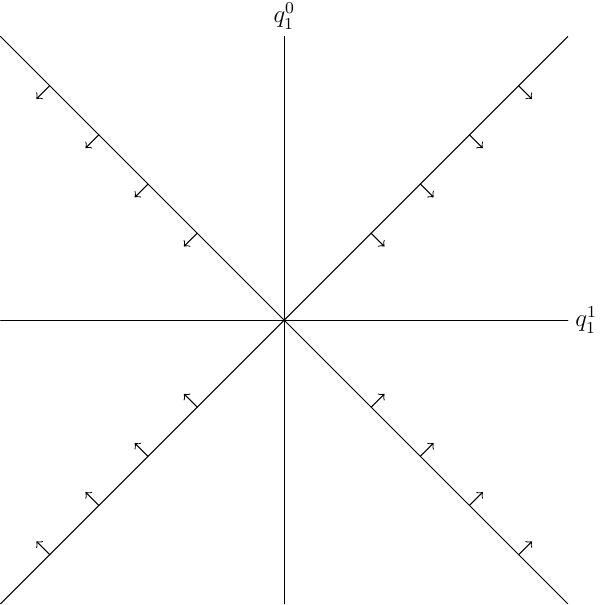}
  \caption{The on-shell surface of a massless particle, together with a choice of complex detour which cannot avoid all the complex points of the singular on-shell hypersurface; there is a pinch at zero momentum.}
  \label{fig:eyelash-massless}
\end{figure}

Pinches which occur for all value of the external kinematics are called \emph{permanent pinches}. When they are present, the integrand needs to vanish sufficiently rapidly so that the integral remains well defined even if the contour passes through the integrand singularity at the pinch. When the integrand does not vanish fast enough, these pinches contribute to an IR divergence and the integral must be properly regulated. Feynman integrals always have permanent pinches at zero momentum for massless particles. Nevertheless, kinematic singularities can be hidden in such propagators, and we may need to make a change of variables to separate the IR divergences from singularities that occur as the external kinematics are varied. In the mathematics literature, such changes of variables are called \emph{blow-ups}. In addition to exposing new singularities, blow-ups can modify the behavior of the integral around the singularity in an asymptotic expansion, as we will discuss below.  Moreover the \emph{hierarchical principle}, which dictates which discontinuities of $\I_G$ are allowed in sequence, has to be reassessed in the presence of massless particles.

We now illustrate how to carry out the blow-up needed to resolve the singularities at zero momentum. We will see that these lead to a new form of the Landau equations, which we solve in two cases---for the bubble singularity with one massless edge, and the sunrise singularity when one of the participating edges is massless.

\subsection{Landau equations for massless particles in loop-momentum space}
\label{sec:landau_massless}

As a warm up, we first consider the bubble integral in $D$ dimensions with one propagator of mass $m \neq 0$ and one massless propagator:
\begin{equation}
    \I_{\text{bub}} = \int \frac {d^D q_1}{q_1^2 (q_2^2 - m^2)},
    \label{eq:bubble}
\end{equation}
with $q_1 + q_2 = p$. We can try to solve the traditional Landau equations, which in this example give the loop equation $\alpha_1 q_1 = \alpha_2 q_2$. When squaring and using the on-shell conditions $q_1^2=0$ and $q_2^2=m^2$, we find $\alpha_1^2 m^2 = 0$, which gives $\alpha_2 = 0$ and thus $q_1 = 0$. 
Notice that this happens for any values of the external kinematics $p$, and is therefore a permanent pinch. It comes about in this case since the Landau singularity and IR divergence of this integral as $q_1 \to 0$ are at the same location in the integration space.

While it is certainly true that the massless bubble integral can be IR divergent (and will be for dimensions $D \leq 2$), it turns out that we have not yet uncovered all the kinematic singularities of this integral. In fact, there is a singularity at $p^2=m^2$ hiding in this corner of integration space as we will see. In order to analyze\footnote{To describe the integration cycle after the blow-up we need to choose a lift to the blown-up space.  We will not describe these choices in detail; one can proceed similarly to ref.~\cite{boyling1968homological} where an analogous construction was done in Feynman parameter space.} it, we do a blow-up around the point at $q_1=0$, which we do by reparametrizing the massless momentum as $q_1 = \rho (1, \vec{y})$. This change of coordinates transforms the light-cone of $q_1$ to a cylinder $\mathbb{R} \times S^{D - 2}$, where $\rho$ is a coordinate on $\mathbb{R}$ and $\vec{y}$ is a coordinate on $S^{D - 2}$.  Indeed, the point at the vertex of the light-cone is replaced by a $D - 2$-sphere given by the equations $\rho = 0$ and $\vec{y}^2 = 1$:

After the blow-up, the integral from~\eqref{eq:bubble} takes the form
\begin{equation}
    \I_{\text{bub}} = \int \frac{\rho^{D-3} d \rho \, d^{D-1} \vec{y}}{(1 - \vec{y}^2) \left[ (p^0 - \rho)^2 - (\vec{p} - \rho \vec{y})^2 - m^2 \right] }.
    \label{eq:massless-bub-blowup}
\end{equation}
This form is better suited for disentangling the IR divergence and the Landau singularity.

For $D = 3$ the differential form we are integrating is a representative of a cohomology class in $H^D(X \setminus S_1 \cup S_2)$, where $X$ is a compactification of $\mathbb{C}^D$, $S_1$ is defined by $\vec{y}^2 = 1$ in $X$ and $S_2$ is defined by $(p^0 - \rho)^2 - (\vec{p} - \rho \vec{y})^2 - m^2 = 0$.  For $D > 3$, we have instead that the differential form is a representative of a relative cohomology class in $H^D(X \setminus S_1 \cup S_2, S_0)$, where $S_0$ is defined by $\rho = 0$.  Finally, for $D < 3$ we have that the differential form is a representative of a cohomology class in $H^D(X \setminus S_1 \cup S_2 \cup S_0)$.

We can rewrite one of the terms in the denominator as
\begin{equation}
  (p^0 - \rho)^2 - (\vec{p} - \rho \vec{y})^2 - m^2 =
  p^2 - m^2 + \rho (-2 p^0 + 2 \vec{p} \cdot \vec{y} + \rho (1 - \vec{y}^2)).
\end{equation}
When $p^2 - m^2 = 0$ we find that the equation factorizes so it has a singularity of type $x y = 0$ (at $x = y = 0$).  This singularity can be resolved as well if we make another change of variable $\rho = (p^2 - m^2) \sigma$.  This produces a global factor of $(p^2 - m^2)^{D - 3}$

To compare the asymptotic expansion of Eq.~\eqref{eq:Ibub_full} when $p^2 - m^2 \to 0$ with the result obtained above we use~\cite[\href{https://dlmf.nist.gov/15.8.E2}{Eq.15.8.2}]{NIST:DLMF} which reads
\begin{equation}
\frac {\sin (\pi (b - a))}\pi {}_2 F_1(a, b; c; z) =
\frac {(-z)^{-a}}{\Gamma(b) \Gamma(c-a)} {}_2 F_1(a, a - c + 1; a - b + 1; z^{-1}) - (a \leftrightarrow b).
\end{equation}
We want to use this for $a = 2 - \frac D 2$, $b = \frac {D - 2} 2$, $c = \frac D 2$, $b - a = D - 3$, $a - c + 1 = 3 - D$, $a - b + 1 = 4 - D$, $b - c + 1 = 0$, $b - a + 1 = D - 2$.

Since
\begin{equation}
{}_2 F_1(a, b; c; z) = 1 + \frac {a b} c z + \frac {a (a + 1) b (b + 1)}{c (c + 1)} \frac {z^2}{2!} + \cdots,
\end{equation}
we have that ${}_2 F_1(a, b; c; w) \sim 1$ as $w \to 0$ (as long as $c$ is not a negative integer).

The two power contributions to the bubble are
\begin{gather}
(m^2 - p^2)^{\frac {D - 4} 2} \bigl(\frac {p^2}{m^2 - p^2}\bigr)^{-2 + \frac D 2} \to (m^2 - p^2)^{0}, \\
(m^2 - p^2)^{\frac {D - 4} 2} \bigl(\frac {p^2}{m^2 - p^2}\bigr)^{-\frac {D-2} 2} \to (m^2 - p^2)^{D - 3}.
\end{gather}
This matches perfectly the analysis above.

The blown-up integration form is also better suited to doing an analysis along the lines of ref.~\cite{Hannesdottir:2022xki}.  There, the Cutkosky formula arose by using Leray's higher-dimensional generalization of Cauchy's formula, which involves constructing a Leray coboundary around an on-shell locus.  However, a Leray coboundary can only be constructed around non-singular points, while the zero momentum of a massless particle is at a singular point of its on-shell space.

As another example, consider the sunrise integral with one massless propagator $m_2 = 0$.  Solving the Landau equations
\begin{equation}
  \alpha_1 q_1 = \alpha_2 q_2 = \alpha_3 q_3,
\end{equation}
together with the on-shell conditions $q_1^2 = m_1^2$, $q_2^2 = 0$, $q_3^2 = m_3^2$, we find $\alpha_1^2 m_1^2 = 0 = \alpha_3^2 m_3^2$.  Therefore, assuming $m_1 \neq 0 \neq m_3$, we have $\alpha_1 = \alpha_3 = 0$ and therefore $\alpha_2 = 1$.  Plugging back into the original equations we find $q_2 = 0$.  This means that the values of the loop momenta for which the Landau singularity occurs are at a singular point in the on-shell space; at the vertex of the light-cone.

We now try to resolve this singularity by blow-up.  We set $q_2 = \rho (1, \vec{y})$ and $\alpha_1 = \rho \beta_1$, $\alpha_3 = \rho \beta_3$.  Here we have in mind that $\rho \to 0$ at the singularity.  In these new coordinates, and factoring out the factor of $\rho$ (as required to compute what is known as the strict transform in mathematical language), we find that the Landau equations become
\begin{equation}
  \beta_1 q_1 = \alpha_2 (1, \vec{y}) = \beta_3 q_3.
\end{equation}

We will try to re-derive the Landau equations by doing a change of variable in the Feynman-parametrized integral inspired by the blow-up change of coordinates above.  The original integral is
\begin{equation}
  2 \int \frac {d^D q_1 d^D q_2 d \alpha_1 d \alpha_3}{\bigl[\alpha_1 (q_1^2 - m_1^2) + \alpha_2 q_2^2 + \alpha_3 (q_3^2 - m_3)^2\bigr]^3}.
\end{equation}
Changing variables to $\alpha_1 = \rho \beta_1$, $\alpha_3 = \rho \beta_3$, $q_2 = \rho (1, \vec{y})$, we find
\begin{equation}
  2 \int \frac {d^D q_1 \rho^{D - 1} d \rho d^{D - 1} \vec{y} \rho^2 d \beta_1 d \beta_3}{\rho^3 \bigl[\beta_1 (q_1^2 - m_1^2) + \alpha_2 \rho (1 - \vec{y}^2) + \beta_3 ((p - q_1 - \rho (1, \vec{y}))^2 - m_3^2)\bigr]^3},
\end{equation}
where $\alpha_2 = 1 - \rho (\beta_1 + \beta_3)$.  Note that our choice of change of coordinates allows us to pull out a cubic factor from the denominator.

We can now apply the usual procedure to the denominator.  We take derivatives with respect to $q_1$, $\rho$, $\vec{y}$ and also $\beta_1$, $\beta_3$ and set them to zero and solve the resulting equations.

We find, for the derivatives with respect to $\beta_1$, $\beta_3$, $q_1$, $\rho$, $\vec{y}$, respectively:
\begin{gather}
  q_1^2 - m_1^2 = \rho^2 (1 - \vec{y}^2), \\
  q_3^2 - m_3^2 = \rho^2 (1 - \vec{y}^2), \\
  \beta_1 q_1 + \beta_3 (q_1 - p + \rho (1, \vec{y})) = 0, \\
  (1 - 2 \rho (\beta_1 + \beta_3)) (1 - \vec{y}^2) + 2 \beta_3 (1, \vec{y}) \cdot (q_1 - p + \rho (1, \vec{y})) = 0, \\
  -(1 - \rho (\beta_1 + \beta_3)) \vec{y} + \beta_3 (\vec{p} - \vec{q}_1 - \rho \vec{y}) = 0.
\end{gather}
From the vector components of the second equation above and the last equation we find
\begin{gather}
  \vec{q}_1 + \frac {1 - \rho \beta_1}{\beta_3} \vec{y} = \vec{p}, \\
  \bigl(1 + \frac {\beta_1}{\beta_3}\bigr) \vec{q}_1 + \rho \vec{y} = \vec{p},
\end{gather}
we find
\begin{gather}
  \vec{q}_1 = \vec{p} \beta_3 \frac{(\beta_1 + \beta_3) \rho - 1}{{\left(\beta_1^2 + \beta_1 \beta_3 + \beta_3^2\right)} \rho - \beta_1 - \beta_3} = \vec{p} \Bigl[\frac {\beta_3}{\beta_1 + \beta_3} - \rho \frac {\beta_1 \beta_3^2}{(\beta_1 + \beta_3)^2} + \mathcal{O}(\rho^2)\Bigr], \\
  \vec{y} = -\frac{\beta_1 \beta_3 \vec{p}}{(\beta_1^2 + \beta_1 \beta_3 + \beta_3^2) \rho - \beta_1 - \beta_3} = \vec{p} \Bigl[\frac {\beta_1 \beta_3}{\beta_1 + \beta_3} + \mathcal{O}(\rho)\Bigr], \\
  \vec{q}_3 = \vec{p} \beta_1 \frac{(\beta_1 + \beta_3) \rho - 1}{(\beta_1^2 + \beta_1 \beta_3 + \beta_3^2) \rho - \beta_1 - \beta_3} = \vec{p} \Bigl[\frac {\beta_1}{\beta_1 + \beta_3} - \rho \frac {\beta_1^2 \beta_3}{(\beta_1 + \beta_3)^2} + \mathcal{O}(\rho^2)\Bigr].
\end{gather}
From the zeroth component of the third equation we find
\begin{gather}
  q_1^0 = \frac {\beta_3}{\beta_1 + \beta_3} (p^0 - \rho), \qquad
  q_3^0 = \frac {\beta_1}{\beta_1 + \beta_3} (p^0 - \rho),
\end{gather}
where we have used momentum conservation for the second equation.  At this point we have solved for the loop momenta in terms of $\rho$, $\beta_1$ and $\beta_3$.  Notice that while $\vec{q}_1$, $\vec{y}$ and $\vec{q}_3$ are proportional to $\vec{p}$, this is not true for the zeroth components.

Next, we use the first two and the fourth equation to solve for the remaining unknowns.  Unfortunately the equations are pretty complicated.\footnote{See ref.~\cite{Landshoff1966} for a blow-up approach in Feynman parameter space, where one also encounters complications in solving the resulting equations.}  However, without solving the equations we have that
\begin{equation}
  \beta_1 (q_1^2 - m_1^2) + \alpha_2 \rho (1 - \vec{y}^2) + \beta_3 (q_3^2 - m_3^2) =
  \rho (1 - \vec{y}^2),
\end{equation}
where we have used $q_1^2 - m_1^2 = \rho^2 (1 - \vec{y}^2) = q_2^2 - m_2^2$ and $\alpha_2 = 1 - \rho (\beta_1 + \beta_3)$.  In this case the exponent of $\rho$ in the integrand is $\rho^{D - 5}$, so the correct answer can not be obtained without an analysis of the Hessian.  We leave this for future work.

\subsection{Alternative Parametrization of the Blow-Up}
\label{sec:conifold}

Above, we have encountered the singularity at the tip of a light-cone and we have resolved it in an ad-hoc way.  In this section we present an alternative parametrization of the blow-up, which is in a sense more covariant and which points to a potential connection to momentum twistors.

The most intuitive way to resolve the massless singularity is to make the massless propagators massive and expect that the massless limit can be controlled and it gives sensible answers.  Another way is to blow it up.  These two approaches lead to \emph{different} results for the resolved space, as described below.

Ref.~\cite[sec.~6.6]{MR2003030} is a source for the discussion in this section.  The complexified on-shell equation for a massless particle in four space-time dimensions can be written as
\begin{equation}
A^2 + B^2 + C^2 + D^2 = 0,
\end{equation}
where $A, B, C, D \in \mathbb{C}$.  Since we complexify, the signature of the metric is not a good notion anymore.

If we introduce new coordinates
\begin{equation}
X = A + i B, \qquad
Y = A - i B, \qquad
U = D + i C, \qquad
V = -D + i C,
\end{equation}
then the equation for the massless on-shell space becomes $X Y - U V = 0$,  or
\begin{equation}
\det \begin{pmatrix}
X & U \\
V & Y
\end{pmatrix} = 0.
\end{equation}

In terms of real coordinates
\begin{equation}
\vec{x} = (\Re A, \Re B, \Re C, \Re D), \qquad
\vec{y} = (\Im A, \Im B, \Im C, \Im D),
\end{equation}
we have the equations $(\vec{x} + i \vec{y})^2 = 0$ or, in real and imaginary parts
\begin{equation}
\vec{x}^2 - \vec{y}^2 = 0, \qquad
\vec{x} \cdot \vec{y} = 0.
\end{equation}
If we set $\vec{x}^2 = r^2$ then, at fixed $r$ we have that $\vec{x}$ belongs to an $S^3$, while $\vec{y}$, being orthogonal to $\vec{x}$, belongs to an $S^2$.  If we let $r$ range over $r \in [0, \infty)$, we obtain a cone over $S^3 \times S^2$.  This is why this is called a \emph{conifold singularity}.

If we add a mass, we obtain the equation
\begin{equation}
A^2 + B^2 + C^2 + D^2 = m^2.
\end{equation}
If the mass is real, then we have
\begin{equation}
\vec{x}^2 - \vec{y}^2 = m^2, \qquad
\vec{x} \cdot \vec{y} = 0.
\end{equation}
Now we can introduce $r^2 = \vec{y}^2$ and $\vec{x}^2 = r^2 + m^2$ and $r \in [0, \infty)$.  At fixed $r$ we again have a $S^3 \times S^2$ but at $r = 0$ the $S^2$ sphere shrinks to zero size while the $S^3$ has minimal radius $m$.  So this complex structure deformation has replaced the singular tip of the conifold by an $S^3$.

If we don't want to introduce any masses, then an option is to use a blow-up.  It is best to use the coordinates $X, Y, U, V$ for this purpose.  As should be very familiar, when $X Y - U V = 0$, then we can find $\lambda_1, \lambda_2 \in \mathbb{C}$, not both vanishing, such that
\begin{equation}
\begin{pmatrix}
X & U \\
V & Y
\end{pmatrix}
\begin{pmatrix}
\lambda_1 \\ \lambda_2
\end{pmatrix}  = 0.
\end{equation}
As long as not all the coordinates $X, Y, U, V$ vanish, then there is a unique solution for $(\lambda_1 \colon \lambda_2) \in \mathbb{P}^1$, as a coordinate on a projective line.  But if $X, Y, U, V$ vanish, then the $\lambda$ are unconstrained.  This means that the blow-up has added a (exceptional divisor) $\mathbb{P}^1$ at the origin.

There is also a blow-up by action on the left, instead of the action on the right described above.  These equations and the choice of ``chirality'' inherent in this parametrization of the blow-up bear a striking similarity to momentum twistor constructions.  We leave a more detailed investigation of this connection for the future.

We note that the blow-up added a $S^2$, which is a complex manifold while the mass (or complex structure) deformation added an $S^3$ which is not.  This could be useful when studying the integration space by stratification (see ref.~\cite{pham2011singularities}).

\section{Asymptotic Analysis with Massless Particles}
\label{sec:asymptotic}

In this section, we are interested in the expansion of Feynman integrals around Landau singularities. For a singularity at $\varphi=0$ (we could e.g.\ take~$\varphi=p^2-m^2$ in the example from Sec.~\ref{sec:landau_massless}), we expect the leading non-analytic term in an expansion around $\varphi=0$ to take the form
\begin{equation}
    \lim_{\varphi \to 0} \I \sim C \varphi^a \log^b \varphi + \ldots \,.
    \label{eq:expansion}
\end{equation}
Here, $C$ is a constant, and the $\ldots$ denote terms are either higher order in $\varphi$, or are non-singular as $\varphi \to 0$. Note in particular that the term displayed on the right-hand side of~\eqref{eq:expansion} does not need to be the leading term in a $\varphi \to 0$ expansion -- it is only the leading non-analytic term. We call $a$ the \emph{Landau exponent} corresponding to the Feynman integral $\I$.

\subsection{Review: Expansions around generic-mass Landau varieties}
\label{sec:generic_asymptotics}

In the case of Feynman integrals corresponding to one-vertex irreducible Feynman diagrams with sufficiently generic masses, the expansion has long been known to be~\cite{Landau:1959fi,pham2011singularities}
\begin{equation}
    \lim_{\varphi \to 0}\I \sim \begin{cases}
   C  \varphi^\gamma \log \varphi   & \text{if } \gamma \in \mathbb{Z}, \gamma \ge 0 \\ 
   C \varphi^\gamma & \text{otherwise.}
    \end{cases}
\label{eq:I_expansion}
\end{equation}
where $\gamma$ depends on the topology of the diagram and the spacetime dimension $D$.
In other words, we have that $a=\gamma$ and, additionally, $b=1$ if $\gamma$ is a positive integer or zero, and otherwise $b=0$.

There are multiple ways to derive the formula in~\eqref{eq:I_expansion}. One can, work directly in loop-momentum space~\cite{pham1968singularities}, directly in Feynman-parameter space~\cite{PolkinghorneScreaton,Hannesdottir:2022bmo}, or in a mixed representation with Feynman parameters and loop momenta as in Landau's original derivation~\cite{Landau:1959fi}. We take a generic scalar Feynman integral, i.e.~\eqref{eq:I_int} with  with $N=1$. The Landau equations read
\begin{subequations}
    \begin{align}
    \alpha_e (q_e^2-m_e^2) & = 0 \,, \qquad \text{for } e \in \{1,2, \ldots, \E \} \,,
    \label{eq:loop_landau_1}
    \\
    \sum_{\text{loop}} \alpha_e q_e^\mu & = 0 \,
    \qquad \text{for every loop.}
    \label{eq:loop_landau_2}
\end{align}
\end{subequations}
We assume that $\varphi=0$ is a solution of these equations for which the corresponding points in loop momentum space where the pinch occurs are at $\alpha_e = \alpha_e^\ast$ for every $e$ and $k_{\ell} = k_{\ell}^\ast$ for every $\ell \in \mathcal{D}$, and furthermore
\begin{subequations}
    \begin{align}
    q_e^2 - m_e^2 & = 0 \qquad \text{for } e \in \mathcal{C} \,, \\
    q_e^2 - m_e^2 = (q_e^\ast)^2 - m_e^2 & \neq 0 \qquad \text{for } e \not\in \mathcal{C}
\end{align}
\end{subequations}
with $\mathcal{C}$ being a the set of the edges that is defined by these equations. We assume that the remaining loop momenta $k_\ell \not\in \mathcal{D}$ are not fixed by the Landau equations and must instead be integrated over. We define $l=|\mathcal{D}|$ to be the number of loop-momenta that are determined by the Landau equations.

Before we continue, let us emphasize that we have made an important distinction between the denominator factors $q_e^2-m_e^2$ that are set to zero on the singularity and those that are not. Looking back at~\eqref{eq:loop_landau_1}, we could say that we have solved the Landau equations setting $\alpha_{e'}=0$ for $e' \not\in \mathcal{C}$. In summary, we have
\begin{itemize}
    \item $l D$ coordinates that are pinched at the singularity $\varphi=0$
    \begin{equation}
        k_\ell \qquad \text{for $\ell \in \mathcal{D}$}
    \end{equation}
    \item $n$ surfaces that are singular at the singularity $\varphi=0$:
    \begin{equation}
        q_e^2-m_e^2=0 \qquad \text{for $e \in \mathcal{C}$}
    \end{equation}
\end{itemize}

However, as we will see shortly, it is more useful for us to Feynman-parametrize the integral $\I$ omitting the parameters that are not in $\mathcal{C}$ and write
\begin{equation}
    \I = \int_0^\infty \prod_{e \in \mathcal{C}} d \alpha_e \, \delta \big(1-\sum_{e \in \mathcal{C}} \alpha_e\big) \int_h \prod_{\ell=1}^{\L} d^D k_{\ell} \frac{(n - 1)!}{\left[ \sum_{e \in \mathcal{C}} \alpha_e (q_e^2-m_e^2) \right]^{n}} \frac{1}{\prod_{e \not\in \mathcal{C}} (q_e^2-m_e^2)} \,.
\end{equation}
Here we have denoted the number of on-shell edges with $n \equiv |\mathcal{C}|$. In this formulation, the singular denominator has an expansion of the form
\begin{equation}
    \sum_{e\in \mathcal{C}} \alpha_e (q_e^2-m_e^2)
    =
    \varphi + \frac 1 2 \sum_{\alpha,\beta = 1}^{m} H_{\alpha,\beta} x^{\alpha} x^{\beta}  + \ldots \,,
\end{equation}
where $x^i$ are components of the vector $\vec{x}$ which collects the $m \equiv l D+n-1$ coordinates corresponding to $\delta k_\ell=k_\ell -k_\ell^\ast$ and $\delta \alpha_e = \alpha_e - \alpha_e^\ast$ that are determined by the Landau-equation solution. The Hessian is explicitly given by
\begin{equation}
    H_{\alpha,\beta} = \frac{\partial}{\partial x_\alpha} \frac{\partial}{\partial x_\beta} \sum_{e \in \mathcal{C}} \alpha_e (q_e^2-m_e^2) \Bigg\vert_{\varphi=0, \vec{x}=0} \,.
\end{equation}
The expansion is valid close to $\vec{x}=\vec{0}$.  Higher order terms yield subleading contributions.  The leading contribution to the integral is then\footnote{This assumes that all the $\alpha_e^* > 0$.  If some of them vanish, the corresponding integral is over an interval $[0, \delta)$ instead.}
\begin{equation}
    \I \sim \int_{-\delta}^{\delta} d^{m} \vec{x}
    \int \prod_{\ell \not\in \mathcal{D}} d^D k_{\ell}
    \frac{(n - 1)!}{\left[\varphi + \frac 1 2 \sum_{\alpha,\beta = 1}^{m} H_{\alpha,\beta} x^{\alpha} x^{\beta}\right]^n} \frac{1}{\prod_{e \not\in \mathcal{C}} [(q_e^\ast)^2-m_e^2]} \,.
\end{equation}
We can easily perform the integral over $\vec{x}$, under the assumptions that the quadric in $\vec{x}$ can be diagonalized, i.e. $\det H_{\alpha \beta} \neq 0$, and that the remaining integrals do not strengthen or weaken the singularity. Then, we get
\begin{equation}
    \mathcal{I} \sim C \varphi^{\gamma} \times
    \begin{cases}
    \varphi^{-\gamma} + \varphi^{-\gamma+1} + \cdots + \varphi^{-1}
    +
    \log \varphi + \cdots & \text{if } \gamma \in \mathbb{Z}_+,
    \\ \log \varphi & \text{if } \gamma = 0 \\ 
    1 + \cdots & \text{otherwise,}
    \end{cases}
    \label{eq:Landau_exp_derivation}
\end{equation}
where $C$ involves the remaining integrations over $k_\ell$ for $\ell \not \in \mathcal{D}$. Going back to~\eqref{eq:I_expansion},we see that~\eqref{eq:Landau_exp_derivation} is consistent provided that we identify
\begin{equation}
    \gamma = \frac{l D -n-1}{2} \,,
    \label{eq:Landau_exponent_gamma}
\end{equation}
where $l$ is the number of loop momenta that are pinched and $n$ is the number of on-shell propagators in the Landau-equation solution for the singularity at $\varphi=0$.

When massless particles are involved, however, the generic-mass formula from~\eqref{eq:I_expansion} and~\eqref{eq:Landau_exponent_gamma} does not always apply. As a simple example, the singularity of the one-mass bubble integral from~\eqref{eq:Ibub_full} at $p^2-m^2=0$ has $\gamma=\frac{D-4}{2}$, while na\"ively applying~\eqref{eq:Landau_exponent_gamma} would incorrectly predict that $\gamma=\frac{D-3}{2}$. This discrepancy is of course entirely expected: the Eq.~\eqref{eq:Landau_exponent_gamma} was derived under the assumptions that the kinematics and masses were sufficiently generic.
However, an analogous derivation goes through if we take into account the resolution of singularities needed to obtain a simple pinch.  We now explore this method in some examples.

\subsection{Massless bubble}
\label{sec:massless-bubble}

In this section we compute the Landau exponents for the massless bubble.
The idea involves going below the threshold and doing the analysis of Polkinghorne \& Screaton (see ref.~\cite{PolkinghorneScreaton}) then noticing that the Hessian itself contributes a singular factor.  This is how the non-simple pinch nature of the singularity manifests itself in this approach to the problem.

We present a detailed calculation in the massless bubble case, but the same method should work in all cases where the Hessian determinant is not exactly zero.

First, let us solve the usual Landau equations.  We have $p = q_1 + q_2$, the Landau loop equation $\alpha_1 q_1 = \alpha_2 q_2$ and the on-shell equations $q_1^2 = q_2^2$.  So we get
\begin{gather}
q_1 = \frac {\alpha_2}{\alpha_1 + \alpha_2} p, \qquad
q_2 = \frac {\alpha_1}{\alpha_1 + \alpha_2} p.
\end{gather}
This implies $p^2 = 0$, but $\alpha_1, \alpha_2$ are not determined.  We do not have a simple pinch.

Now let us repeat this analysis using the Polkinghorne \& Screaton idea.  We take $p^2 = -\epsilon$ for $\epsilon > 0$.  The original integral can be written in a Feynman-parametrized form
\begin{equation}
\int \frac {d^d q_1}{q_1^2 q_2^2} = \int \frac {d \alpha_1 d^d q_1}{\Bigl(\alpha_1 q_1^2 + (1 - \alpha_1) q_2^2\Bigr)^2}.
\end{equation}
Next we analyze the function
\begin{equation}
F(q_1, \alpha_1) = \alpha_1 q_1^2 + (1 - \alpha_1) (p - q_1)^2.
\end{equation}
Its stationary points conditions are
\begin{gather}
0 = \frac {\partial F}{\partial \alpha_1} = q_1^2 - (p - q_1)^2, \\
0 = \frac {\partial F}{\partial q_1} = 2 q_1 - 2 (1 - \alpha_1) p.
\end{gather}
Therefore, we have $q_1^* = (1 - \alpha_1) p$, $q_2^* = \alpha_1 p$.

Using $q_1^2 - q_2^2 = 0$ we find $(1 - 2 \alpha_1) p^2 = 0$.  Since $p^2 = -\epsilon \neq 0$ we have $\alpha_1^* = \alpha_2^* = \frac 1 2$.  This is a \emph{different} solution than the one obtained by solving the usual Landau equations.  In particular, it \emph{is} a simple pinch.

At the stationary point $(q_1, \alpha_1) = (q_1^*, \alpha_1^*)$ we have
\begin{equation}
F^* = F(q_1^*, \alpha_1^*) = \frac {p^2} 4 = -\frac \epsilon 4.
\end{equation}

Next we expand $F$ around the stationary point.  We have
\begin{equation}
F(q_1, \alpha_1) = -\frac \epsilon 4 + (q_1 - q_1^*)^2 + p \cdot (q_1 - q_1^*) (\alpha_1 - \alpha_1^*) + \cdots,
\end{equation}
where the second order terms arise from the following $(D + 1) \times (D + 1)$ Hessian matrix
\begin{equation}
H = \begin{pmatrix}
\eta_{\mu \nu} & \frac {p_\mu} 2 \\
\frac {p_\nu} 2 & 0
\end{pmatrix}.
\end{equation}

We can compute $\det H$ by dotting the first $D$ rows with $\frac {p_\mu} 2$ and subtracting the answer from the last row.  This cancels the lower-left off-diagonal term rendering the Hessian matrix block upper diagonal.  The determinant is $(-1)^D \frac {p^2} 4 = (-1)^{D-1} \frac \epsilon 4$ in Lorentzian mostly minus signature.  When $\epsilon \to 0$ this Hessian becomes degenerate, but we are doing this study at small but non-vanishing $\epsilon$.

Changing coordinates to
\begin{gather}
\alpha_1 = \alpha_1^* + \sqrt{\epsilon} \beta, \\
q_1 = q_1^* + \sqrt{\epsilon} \xi,
\end{gather}
we find the integral
\begin{equation}
\int \frac {\epsilon^{\frac {D + 1} 2} d \beta d^d \xi}{\Bigl(-\frac \epsilon 2 + \epsilon \xi^2 + \epsilon p \cdot \xi \beta\Bigr)^2}.
\end{equation}
In the limit $\epsilon \to 0$ and dropping numerical factors, this integral behaves as
\begin{equation}
\frac {\epsilon^{\frac {D - 3} 2}}{\sqrt{\epsilon}} = \epsilon^{\frac {D - 4} 2},
\end{equation}
where the square root arises from the square root of the determinant of the Hessian matrix (the analysis if the same as in ref.~\cite{Hannesdottir:2024hke}).

An important characteristic of this mechanism is that it always shifts the Landau exponent \emph{down}.  That is, we are never in danger of exceeding the transcendentality barrier of $\lfloor \frac D 2\rfloor$.

For references, the massless bubble integral in $D$ dimensions is given by
\begin{equation}
    I_{\text{bub}}= \Gamma \left( 2-\frac{D}{2} \right) \int_0^1 d \alpha \frac{1}{[s \alpha (1-\alpha)]^{2-D/2}}
    = \frac{\sqrt{\pi } 2^{3-D} s^{\frac{D-4}{2}} \Gamma \left(2-\frac{D}{2}\right) \Gamma \left(\frac{D}{2}-1\right)}{\Gamma \left(\frac{D-1}{2}\right)} \,.
\end{equation}
Here we have used the Euler Beta integral $\int_0^1 x^{p - 1} (1 - x)^{q - 1} d x = B(p, q)$, with $B(p, q) = \frac {\Gamma(p) \Gamma(q)}{\Gamma(p + q)}$.  Then, we have applied the duplication formula $\Gamma(z) \Gamma(z + \frac 1 2) = 2^{1 - 2 z} \sqrt{\pi} \Gamma(2 z)$ for $z = \frac D 2 - 1$.

The expansion around 2, 3 and 4 dimensions (setting $D=4-2\epsilon$ and expanding around $\epsilon=1$, $\epsilon=1/2$ and $\epsilon=0$, respectively), gives
\begin{align}
    I_{\text{bub}}^{2D} & = \frac{2}{s (\epsilon -1)}-\frac{2 (\log (-s)+\gamma )}{s}+O(\epsilon - 1), \\
    I_{\text{bub}}^{3D} & = \frac{\pi ^{3/2}}{\sqrt{-s}}+O\left(\epsilon -\frac{1}{2}\right), \\
    I_{\text{bub}}^{4D} & = \frac{1}{\epsilon }+\left(\log \left(-\frac{4}{s}\right)+\psi ^{(0)}\left(\frac{3}{2}\right)\right)+O\left(\epsilon\right) \,.
\end{align}

\subsection{One-mass bubbles}

Next, we look at the one-mass bubble integral where we set the mass to 1 for simplicity,
\begin{multline}
  \label{eq:Ibub_full}
    I_{\text{bub, 1m}}= \Gamma \left( 2-\frac{D}{2} \right) \int_0^1 d \alpha \frac{1}{[s \alpha (1-\alpha)-\alpha]^{2-D/2}}
    \\ = \frac{2 (1-s)^{\frac{D-4}{2}} \Gamma \left(2-\frac{D}{2}\right) \, _2F_1\left(2-\frac{D}{2},\frac{D-2}{2};\frac{D}{2};\frac{s}{s-1}\right)}{D-2} \,.
\end{multline}
The expansions are now given by
\begin{align}
    I_{\text{bub, 1m}}^{2D} & = \frac{1}{(s-1) (\epsilon -1)}+\frac{\log \left(\frac{1}{(s-1)^2}\right)-\gamma }{s-1}+O(\epsilon -1) \label{eq:one-mass-expansion-2d}, \\
    I_{\text{bub, 1m}}^{4D} & = \frac{1}{\epsilon }+\left(\left(\frac{1}{s}-1\right) \log (1-s)-\gamma +2\right)+O\left(\epsilon\right) \label{eq:one-mass-expansion-4d}\,.
\end{align}

This indicates that the asymptotic expansion around $s = 1$ should be $(s-1)^{D - 3}$.  This is very different from the case of two massless particles, where it was $s^{\frac {D - 4} 2}$.  Let us see how it can be reproduced by our analysis.

We have $F(q_1, \alpha_1) = \alpha_1 q_1^2 + \alpha_2 (q_2^2 - m_2^2)$, where $\alpha_2 = 1 - \alpha_1$ and $q_2 = p - q_1$.  The stationary point conditions read
\begin{gather}
  \frac {\partial F}{\partial q_1} = 2 q_1 - 2 \alpha_2 p, \\
  \frac {\partial F}{\partial \alpha_1} = q_1^2 - (q_2^2 - m_2^2).
\end{gather}
We set $q_1^2 = q_2^2 - m_2^2 = \epsilon$.  We obtain $q_1 = \alpha_2 p$, $q_2 = \alpha_1 p$ and
\begin{gather}
  \sqrt{\frac {q_1^2}{p^2}} = \sqrt{\frac {\epsilon}{p^2}} = \alpha_2, \qquad
  \sqrt{\frac {q_2^2}{p^2}} = \sqrt{\frac {\epsilon + m_2^2}{p^2}} = \alpha_1.
\end{gather}
Then, $p^2 = (\sqrt{\epsilon} + \sqrt{\epsilon + m_2^2})^2 = m_2^2 + 2 m \sqrt{\epsilon} + \mathcal{O}(\epsilon^2)$.  We can now express $\epsilon$ in terms of external kinematics
\begin{equation}
  \epsilon \sim \Bigl(\frac {p^2 - m_2^2}{2 m_2}\Bigr)^2.
\end{equation}

We compute the Hessian matrix
\begin{gather}
  \frac {\partial^2 F}{\partial q_1 \partial q_1} = 2 \eta, \qquad
  \frac {\partial^2 F}{\partial q_1 \partial \alpha_1} = 2 p,
\end{gather}
which is
\begin{equation}
  H =
  \begin{pmatrix}
    \eta_{\mu \nu} & 2 p_\mu \\
    2 p_\nu & 0
  \end{pmatrix}.
\end{equation}
We have $\det H = (-1)^{D - 1} 4 p^2$.  Unlike in the massless case above it is non-singular so it will not contribute to the shift in the Landau exponent.

Using, as before, the scaling
\begin{gather}
  \alpha_1 = \alpha_1^* + \sqrt{\epsilon} \beta, \qquad
  q_1 = q_1^* + \sqrt{\epsilon} \xi,
\end{gather}
we obtain that the integral behaves as $\epsilon^{\frac {D - 3} 2}$.  However now we have $(p^2 - m_2^2) \sim \sqrt{\epsilon}$ so the behavior is the expected one $(p^2 - m_2^2)^{D - 3}$.  In contrast, in the massless case we had $p^2 \sim \epsilon$.

\subsection{Bubble integral in two dimensions}
\label{sec:massive-bubble-2D}

In this section we comment on how the hierarchy of singularities and the Pham compatibility conditions constrain the answer for the bubble integral.  A related discussion can be found in ref.~\cite{Hannesdottir:2024hke}.

The bubble integral fits in a Pham diagram (see fig.~\ref{fig:bubble-pham-web}) as follows: we will depict it as an elementary vertex.  There are two Landau diagrams that contract to this vertex, which are two tadpoles.  And finally there is a bubble diagram which contracts both to the vertex and to each of the tadpoles.

\begin{figure}
  \centering
  \includegraphics{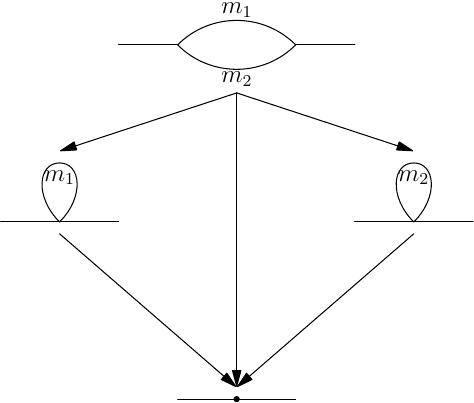}
  \caption{Hierarchy of singularities for a bubble integral.  The dot at the bottom represents the bubble integral itself, while the rest of the diagrams display the cut propagators.}
  \label{fig:bubble-pham-web}
\end{figure}

The tadpole corresponding to the mass $m_1$ has a singularity at $m_1 = 0$ and similarly for $m_2$.  These are logarithmic singularities in two dimensions.  Next, the bubble corresponds to a singularity at $s - (m_1 \pm m_2)^2 = 0$ and it is a square root singularity in two dimensions.

The two tadpoles are not hierarchical to each-other, nor Pham compatible, so it should not be possible to have iterated $m_1 = 0$ followed by $m_2 = 0$ singularities.  Therefore, as a first approximation we have $\log (m_1 m_2)$ which reproduces the logarithmic singularities from the tadpoles.  It is not possible to have $\log m_1 \log m_2$ since then it would be possible to take a monodromy around $m_1 = 0$ followed by a monodromy around $m_2 = 0$ and obtain a non-vanishing result.  It also exceeds the expected transcendentality bound.

However, this does not reproduce the square root singularities.  Let us change the argument of the logarithm to an expression $\frac {P + \sqrt{Q}}{P - \sqrt{Q}}$, where $Q = (s - (m_1 - m_2)^2) (s - (m_1 + m_2)^2)$ and $P$ is determined by the requirement that $P^2 - Q \propto (m_1 m_2)^n$.  By dimensional analysis $n = 2$.  Then it is clear that $P = s - (m_1^2 + m_2^2)$.  And $P^2 - Q = 4 m_1^2 m_2^2$.

Finally, we have $\log \frac {P + \sqrt{Q}}{P - \sqrt{Q}}$ as a candidate answer, but should we pick plus or minus in the numerator?  Also the answer is not dimensionally correct.  To solve both of these problems we take
\begin{equation}
\frac 1 {\sqrt{Q}} \log \frac {P + \sqrt{Q}}{P - \sqrt{Q}},
\end{equation}
for $Q = (s - (m_1 - m_2)^2) (s - (m_1 + m_2)^2)$ and $P = s - (m_1^2 + m_2^2)$.  There are some factors of $\pi$ and other numerical factors which are not fixed yet, but can be fixed from the computation of the maximal cut.

This form of the answer still has some problems.  For example, in this form the threshold and pseudo-threshold occur on an equal footing.  Another issue arises from combining the square roots.

To understand what goes wrong, consider the function
\begin{equation}
  f(z) = \sqrt{(z - a) (z - b)}.
\end{equation}
The quantity under the square root is positive for $z \in (-\infty, a) \cup (b, \infty)$ (where we assume without loss of generality that $a < b$).  It is negative for $z \in (a, b)$, and there we expect a branch cut.  To understand this branch cut, take $z = x + i \epsilon$ for real $x$ and $\epsilon \to 0$.

The quantity under the square root is now
\begin{equation}
  (x + i \epsilon - a)(x + i \epsilon - b) = (x - a)(x - b) + i \epsilon (2 x - a - b) - \epsilon^2.
\end{equation}
It is now easy to see that for $x < \frac {a + b} 2$ the imaginary part is negative, for $x > \frac {a + b} 2$ it is positive and it switches sign at $x = \frac {a + b} 2$.  No such issues arise for the function $\sqrt{z - a} \sqrt{z - b}$.

So we need to split the two square roots as
\begin{equation}
  \sqrt{(s - (m_1 + m_2)^2)(s - (m_1 - m_2)^2)} \to
  \sqrt{s - (m_1 + m_2)^2} \sqrt{s - (m_1 - m_2)^2}.
\end{equation}
Once this is done we notice that the ratio $\frac {P + \sqrt{Q}}{P - \sqrt{Q}}$ is a perfect square.  Extracting the square root and splitting the logarithms we obtain
\begin{multline}
  \frac 1 {\sqrt{s - (m_1 + m_2)^2} \sqrt{s - (m_1 - m_2)^2}} \Bigl(
  \log (\sqrt{s - (m_1 + m_2)^2} - \sqrt{s - (m_1 - m_2)^2}) - \\
  \log (\sqrt{s - (m_1 + m_2)^2} + \sqrt{s - (m_1 - m_2)^2})\Bigr).
\end{multline}

When going around $s = (m_1 - m_2)^2$ we take $s - (m_1 - m_2)^2 \to e^{2 \pi i} (s - (m_1 - m_2)^2)$ which implies $\sqrt{s - (m_1 - m_2)^2} \to -\sqrt{s - (m_1 - m_2)^2}$.  Under this transformation the two logarithms above get swapped and we also obtain a sign from the global prefactor.  The combination is invariant and therefore the pseudo-threshold singularity is invisible on the main sheet.

For the threshold, we take $s - (m_1 + m_2)^2 \to e^{2 \pi i} (s - (m_1 + m_2)^2)$ and therefore $\sqrt{s - (m_1 + m_2)^2} \to -\sqrt{s - (m_1 + m_2)^2}$.  The answer is not invariant under the monodromy anymore and we obtain a branch cut starting at $s = (m_1 + m_2)^2$.

Now that the bubble singularities are correctly reproduced so we should check that the tadpole singularities have not been disturbed.  It turns out that they haven't but for an interesting reason; when taking $m_e^2 \to e^{2 \pi i} m_e^2$ or $m_e \to e^{\pi i} m_e = -m_e$, the logarithmic monodromy works the same way, but now the threshold and pseudo-threshold terms in the prefactor swap places.

Strictly speaking (see also sec.~\ref{sec:massive-bubble-3D}) we should use $\sqrt{m_e^2}$ instead of $m_e$ in the formula above since the variables the integral depends on are $m_e^2$.

\subsection{Bubble  integral in three dimensions}
\label{sec:massive-bubble-3D}

Let us now ask what happens in three dimensions.  There the situation is completely different, since $m_i = 0$ are not logarithmic singularities anymore, but square root.

In this case the answer should be
\begin{equation}
\frac 1 {\sqrt{s}} \log \frac {m_1 + m_2 + \sqrt{s}}{m_1 + m_2 - \sqrt{s}}.
\end{equation}

Now, we only get the pseudo-threshold singularity at $s = (m_1 - m_2)^2$ if we interpret $m_e = \sqrt{m_e^2}$ and analytically continue $m_e \to e^{2 \pi i} m_e$ to change the relative signs of $m_1$ and $m_2$.

Let us first match the square root singularities.  Since we want to have $\sqrt{m_1^2}$ and $\sqrt{m_2^2}$ and we also want there to not be sequential singularities in $\sqrt{m_1^2}$ followed by $\sqrt{m_2^2}$, the first approximation to the answer is $\sqrt{m_1^2} + \sqrt{m_2^2}$.  Next, we want to match the logarithmic branch points at $s = (m_1 + m_2)^2$ and $s = (m_1 - m_2)^2$.  This can be done by taking
\begin{equation}
\log \frac {s - (\sqrt{m_1^2} + \sqrt{m_2^2})^2}{s - (\sqrt{m_1^2} - \sqrt{m_2^2})^2}.
\end{equation}
This would be a great guess if it weren't for the fact that the leading singularity is $\frac 1 {\sqrt{s}}$.  Since $s = 0$ does not occur on the main sheet (it is a second type singularity where $\alpha_1 + \alpha_2 = 0$ so they can not both be positive), this singularity has to be shielded.  So then by Galois symmetry the symbol has to depend on $\sqrt{s}$ as well.  In terms of $\sqrt{s}$ the Landau singularities are at $\sqrt{s} = \pm  (\sqrt{m_1^2} + \sqrt{m_2^2})$ and $\sqrt{s} = \pm (\sqrt{m_1^2} - \sqrt{m_2^2})$.

We want to also avoid the pseudo-threshold so the only option left is
\begin{equation}
\frac 1 {\sqrt{s}} \log \frac {\sqrt{s} + (\sqrt{m_1^2} + \sqrt{m_2^2})}{\sqrt{s} - (\sqrt{m_1^2} + \sqrt{m_2^2})}.
\end{equation}

We could have also noticed another problem that the first ansatz has, which is that after we changed the logarithm argument to match the logarithmic singularities, it is \emph{not true} anymore that the discontinuity under $m_1^2 \to e^{2 \pi i} m_1^2$ does not have a square root singularity at $m_2^2 = 0$.  However, the second form has this property.  Indeed,
\begin{equation}
\log \frac {\sqrt{s} + (\sqrt{m_1^2} + \sqrt{m_2^2})}{\sqrt{s} - (\sqrt{m_1^2} + \sqrt{m_2^2})} \to
\log \frac {\sqrt{s} + (-\sqrt{m_1^2} + \sqrt{m_2^2})}{\sqrt{s} - (-\sqrt{m_1^2} + \sqrt{m_2^2})}
\end{equation}
and the difference is
\begin{equation}
\log \frac {(\sqrt{s} + \sqrt{m_1^2})^2 - m_2^2}{(\sqrt{s} - \sqrt{m_1^2})^2 - m_2^2}.
\end{equation}
Notice how $\sqrt{m_2^2}$ has disappeared from the discontinuity!

The answer can be written in a different way, using
\begin{gather}
L_1 = \frac {(\sqrt{s} + m_1)^2 - m_2^2}{(\sqrt{s} - m_1)^2 - m_2^2} = \frac {(\sqrt{s} + m_1 + m_2)(\sqrt{s} + m_1 - m_2)}{(\sqrt{s} - m_1 + m_2)(\sqrt{s} - m_1 - m_2)}, \\
L_2 = \frac {(\sqrt{s} + m_2)^2 - m_1^2}{(\sqrt{s} - m_2)^2 - m_1^2} = \frac {(\sqrt{s} + m_1 + m_2)(\sqrt{s} - m_1 + m_2)}{(\sqrt{s} + m_1 - m_2)(\sqrt{s} - m_1 - m_2)}.
\end{gather}
Taking the product we find
\begin{equation}
L_1 L_2 = \Bigl(\frac {\sqrt{s} + m_1 + m_2}{\sqrt{s} - m_1 - m_2}\Bigr)^2,
\end{equation}
which is the square of the argument of the logarithm in the expression above for the bubble integral in three dimensions.

An advantage of this writing is that the arguments of the logarithm have a simpler structure, with a single square root.
\begin{gather}
L_1 = \frac{s + m_1^2 - m_2^2 + 2 m_1 \sqrt{s}}{s + m_1^2 - m_2^2 - 2 m_1 \sqrt{s}}, \\
L_2 = \frac{s - m_1^2 + m_2^2 + 2 m_2 \sqrt{s}}{s - m_1^2 + m_2^2 - 2 m_2 \sqrt{s}}.
\end{gather}
This form could have been obtained by starting with $\sqrt{m_1^2}$ and trying to find a combination $L_1 = \frac {P_1 + \sqrt{Q_1}}{P_1 - \sqrt{Q_1}}$ where $Q_1$ contains only $\sqrt{m_1^2}$ and $\sqrt{s}$ and similarly for $\sqrt{m_2^2}$ and $L_2$.

\section{Sunrise}

\subsection{Sunrise with three massive propagators}
\label{sec:sunrise-three-massive}

We first analyze the case with three massive propagators.  We will be able to obtain the cases with massless propagators by setting some of the masses to zero.  We have an integral
\begin{equation}
  2 \int \frac {d^d q_1 d^d q_2 d \alpha_1 d \alpha_2}{F^3},
\end{equation}
where
\begin{equation}
  F(q_1, q_2, \alpha_1, \alpha_2) = \alpha_1 (q_1^2 - m_1^2) + \alpha_2 (q_2^2 - m_2^2) + \alpha_3 (q_3^2 - m_3^2)
\end{equation}
where $q_3 = p - q_1 - q_2$ and $\alpha_3 = 1 - \alpha_1 - \alpha_2$.

Imposing the stationary point conditions we find
\begin{gather}
  0 = \frac {\partial F}{\partial \alpha_1} = q_1^2 - m_1^2 - (q_3^2 - m_3^2), \\
  0 = \frac {\partial F}{\partial \alpha_2} = q_2^2 - m_2^2 - (q_3^2 - m_3^2), \\
  0 = \frac {\partial F}{\partial q_1} = 2 (\alpha_1 q_1 + (1 - \alpha_1 - \alpha_2) (q_1 + q_2 - p)), \\
  0 = \frac {\partial F}{\partial q_2} = 2 (\alpha_2 q_2 + (1 - \alpha_1 - \alpha_2) (q_1 + q_2 - p)).
\end{gather}

The Landau loop equations are the same as before but the on-shell conditions are not.  For this reason we will call them deformed on-shell equations below.

Let us first solve these conditions in the case where we have the usual on-shell conditions $q_i^2 = m_i^2$ for $i = 1, 2, 3$ and $p = q_1 + q_2 + q_3$.  The Landau loop equations read $\alpha_1 q_1 = \alpha_2 q_2 = \alpha_3 q_3 = v$, where $v$ is the common value.  From these equations we have $q_i = \alpha_i^{-1} v$ and from momentum conservation we have
\begin{equation}
  p = q_1 + q_2 + q_3 = (\alpha_1^{-1} + \alpha_2^{-1} + \alpha_3^{-1}) v.
\end{equation}
Therefore
\begin{equation}
  v = \frac p {\alpha_1^{-1} + \alpha_2^{-1} + \alpha_3^{-1}}
\end{equation}
and
\begin{equation}
  \label{eq:landau-loop-sol}
  q_i = \frac {p \alpha_i^{-1}}  {\alpha_1^{-1} + \alpha_2^{-1} + \alpha_3^{-1}}
\end{equation}

At the threshold, where $p^2 = (m_1 + m_2 + m_3)^2$ we have $\alpha_i = \frac {m_i^{-1}}{\sum_j m_j^{-1}}$.  Explicitly, we have
\begin{gather}
  \alpha_1 = \frac {m_2 m_3}{m_1 m_2 + m_1 m_3 + m_2 m_3}, \\
  \alpha_2 = \frac {m_1 m_3}{m_1 m_2 + m_1 m_3 + m_2 m_3}, \\
  \alpha_3 = \frac {m_1 m_2}{m_1 m_2 + m_1 m_3 + m_2 m_3}.
\end{gather}
We will study cases where some $m_e = 0$ below.

Instead of imposing the on-shell conditions, let us now impose
\begin{equation}
  q_1^2 - m_1^2 = q_2^2 - m_2^2 = q_3^2 - m_3^2 = \epsilon,
\end{equation}
where $\epsilon$ is a small quantity.\footnote{We could have chosen to parametrize the neighborhood of the threshold by $p^2 = (m_1 + m_2 + m_3)^2 - \epsilon$, but our choice of parametrization has several benefits.  In particular, we have, at the stationary point, that $F = (\alpha_1 + \alpha_2 + \alpha_3) \epsilon = \epsilon$.}  The solutions of Landau loop equations are the same as in Eq.~\eqref{eq:landau-loop-sol}.  

We have
\begin{gather}
  \label{eq:alpha-relations}
  \frac {\alpha_1^{-1}}{\alpha_1^{-1} + \alpha_2^{-1} + (1 - \alpha_1 - \alpha_2)^{-1}} = \sqrt{\frac {q_1^2}{p^2}} = \sqrt{\frac{\epsilon + m_1^2}{p^2}}, \\
  \frac {\alpha_2^{-1}}{\alpha_1^{-1} + \alpha_2^{-1} + (1 - \alpha_1 - \alpha_2)^{-1}} = \sqrt{\frac {q_2^2}{p^2}} = \sqrt{\frac{\epsilon + m_2^2}{p^2}}, \\
  \frac {(1 - \alpha_1 - \alpha_2)^{-1}}{\alpha_1^{-1} + \alpha_2^{-1} + (1 - \alpha_1 - \alpha_2)^{-1}} = \sqrt{\frac {q_3^2}{p^2}} = \sqrt{\frac{\epsilon + m_3^2}{p^2}}.
\end{gather}
Adding them up we find
\begin{multline}
  p^2 = \Bigl(\sqrt{\epsilon + m_1^2} + \sqrt{\epsilon + m_2^2} + \sqrt{\epsilon + m_3^2}\Bigr)^2 = \\
  = (m_1 + m_2 + m_3)^2 + \epsilon (m_1 + m_2 + m_3) (m_1^{-1} + m_2^{-1} + m_3^{-1}) + \mathcal{O}(\epsilon^2).
\end{multline}
In this expansion in powers of $\epsilon$ it is crucial that none of the masses vanish.  When they do (see below) the expansion becomes an expansion in $\sqrt{\epsilon}$ instead.  This is one of the ways the discontinuous nature of the limit $\epsilon \to 0$ manifests itself.

Then, it follows that
\begin{equation}
  \label{eq:solutions_alpha}
  \alpha_e^* = \frac {(\epsilon + m_e^2)^{-\frac 1 2}}{(\epsilon + m_1^2)^{-\frac 1 2} + (\epsilon + m_2^2)^{-\frac 1 2} + (\epsilon + m_3^2)^{-\frac 1 2}}
\end{equation}
and
\begin{equation}
  \label{eq:solutions_q}
  q_e^* = \frac {\sqrt{\epsilon + m_e^2}}{\sqrt{\epsilon + m_1^2} + \sqrt{\epsilon + m_2^2} + \sqrt{\epsilon + m_3^2}} p.
\end{equation}

The value of $F$ at the critical point is
\begin{equation}
F^* = \alpha_1^* ((q_1^*)^2 - m_1^2) + \alpha_2^* ((q_2^*)^2 - m_2^2) + \alpha_3^* ((q_3^*)^2 - m_3^2) = (\alpha_1^* + \alpha_2^* + \alpha_3^*) \epsilon = \epsilon.
\end{equation}

%%%%%%%%%%%%%%%%%%%%%%%%%%%%%%%%%%%%%%%%%%%%%%%%%%%%%%%%%%%%%%%%%%%%%%
For the Hessian matrix we need the second order derivatives
\begin{gather}
  \frac {\partial^2 F}{\partial q_1^\mu \partial q_1^\nu} =
  2 (1 - \alpha_2) \eta_{\mu \nu}, \\
  \frac {\partial^2 F}{\partial q_1^\mu \partial q_2^\nu} =
  2 (1 - \alpha_1 - \alpha_2) \eta_{\mu \nu}, \\
  \frac {\partial^2 F}{\partial q_2^\mu \partial q_2^\nu} =
  2 (1 - \alpha_1) \eta_{\mu \nu}, \\
  \frac {\partial^2 F}{\partial q_1^\mu \partial \alpha_1} = 2 (p - q_2)_\mu, \\
  \frac {\partial^2 F}{\partial q_1^\mu \partial \alpha_2} = 2 (p - q_1 - q_2)_\mu = 2 q_{3, \mu}, \\
  \frac {\partial^2 F}{\partial q_2^\mu \partial \alpha_1} = 2 (p - q_1 - q_2)_\mu = 2 q_{3, \mu}, \\
  \frac {\partial^2 F}{\partial q_2^\mu \partial \alpha_2} = 2 (p - q_1)_\mu.
\end{gather}
Hence, the Hessian matrix is
\begin{equation}
  \label{eq:hessian}
  H = 2 \begin{pmatrix}
(1 - \alpha_2) \eta_{\mu \nu} & (1 - \alpha_1 - \alpha_2) \eta_{\mu \nu} & (p - q_2)_\nu & q_{3, \nu}, \\
(1 - \alpha_1 - \alpha_2) \eta_{\mu \nu} & (1 - \alpha_1) \eta_{\mu \nu} & q_{3, \nu} & (p - q_1)_\nu \\
(p - q_2)_{\mu} & q_{3, \mu} & 0 & 0 \\
q_{3, \mu} & (p - q_1)_{\mu} & 0 & 0
\end{pmatrix} = \\
2 \begin{pmatrix}
A & B \\
B^t & 0
\end{pmatrix},
\end{equation}
where
\begin{gather}
  A =
  \begin{pmatrix}
    (1 - \alpha_2) \eta_{\mu \nu} & (1 - \alpha_1 - \alpha_2) \eta_{\mu \nu} \\
    (1 - \alpha_1 - \alpha_2) \eta_{\mu \nu} & (1 - \alpha_1) \eta_{\mu \nu}
  \end{pmatrix}, \\
  B =
  \begin{pmatrix}
    (p - q_2)_\nu & q_{3, \nu} \\
    q_{3, \nu} & (p - q_1)_\nu
  \end{pmatrix}.
\end{gather}
The matrix $A$ is $(2 D) \times (2 D)$-dimensional while the matrix $B$ is $(2 D) \times 2$-dimensional.

Next, we use the fact that
\begin{equation}
\det \begin{pmatrix}
A & B \\ B^t & 0
\end{pmatrix} =
\det \begin{pmatrix}
A & B \\ 0 & -B^t A^{-1} B
\end{pmatrix} =
\det A \det (-B^t A^{-1} B).
\end{equation}
Since
\begin{equation}
A = \begin{pmatrix}
(1 - \alpha_2) \eta & (1 - \alpha_1 - \alpha_2) \eta \\
(1 - \alpha_1 - \alpha_2) \eta & (1 - \alpha_1) \eta
\end{pmatrix} = \\
\begin{pmatrix}
1 - \alpha_2 & (1 - \alpha_1 - \alpha_2) \\
(1 - \alpha_1 - \alpha_2) & 1 - \alpha_1
\end{pmatrix} \otimes \eta.
\end{equation}
Then,
\begin{equation}
A^{-1} = \frac 1 {\alpha_1 \alpha_2 + \alpha_1 \alpha_3 + \alpha_2 \alpha_3}
\begin{pmatrix}
1 - \alpha_1 & \alpha_1 + \alpha_2 - 1\\
\alpha_1 + \alpha_2 - 1 & 1 - \alpha_2
\end{pmatrix} \otimes \eta^{-1}
\end{equation}
and
\begin{equation}
  \det A = \Bigl(\alpha_1 \alpha_2 + \alpha_1 \alpha_3 + \alpha_2 \alpha_3\Bigr)^D (\det \eta)^2,
\end{equation}
where we have used the fact\footnote{This is an easy consequence of $P \otimes Q = (P \otimes 1) (1 \otimes Q)$.} that $\det (P \otimes Q) = (\det P)^{\operatorname{rank} Q} (\det Q)^{\operatorname{rank} P}$.

In the end, $B^t A^{-1} B$ is a $2 \times 2$ matrix which can be computed explicitly.  We find
\begin{multline}
  B^t A^{-1} B =
  \frac 1 {\alpha_1 \alpha_2 + \alpha_1 \alpha_3 + \alpha_2 \alpha_3} \times \\
  \begin{pmatrix}
    p - q_2 & q_3 \\
    q_3 & p - q_1
  \end{pmatrix}
  \begin{pmatrix}
    (1 - \alpha_1) \eta^{-1} & (\alpha_1 + \alpha_2 - 1) \eta^{-1} \\
(\alpha_1 + \alpha_2 - 1) \eta^{-1} & (1 - \alpha_2) \eta^{-1}
  \end{pmatrix}
  \begin{pmatrix}
    p - q_2 & q_3 \\
    q_3 & p - q_1
  \end{pmatrix}.
\end{multline}

The determinant can be computed explicitly, to find
\begin{equation}
  \det ((B^*)^t (A^*)^{-1} B^*) =
  \frac{(\alpha_1^*)^2 (\alpha_2^*)^2 (\alpha_3^*)^2 (p^2)^2}{(\alpha_1^* \alpha_2^* + \alpha_1^* \alpha_3^* + \alpha_2^* \alpha_3^*)^{5}},
\end{equation}
where star means evaluated at the critical point.  Here we have used the solutions of the Landau equations to express $q_e^*$ in terms of $p$.

In the end we have
\begin{equation}
  \label{eq:sunrise-hessian}
  \det H^* = 2^{2 D + 2} (\alpha_1^*)^2 (\alpha_2^*)^2 (\alpha_3^*)^2 (p^2)^2 (\alpha_1^* \alpha_2^* + \alpha_1^* \alpha_3^* + \alpha_2^* \alpha_3^*)^{D - 5}.
\end{equation}
Recall that this is an expansion around $\epsilon = 0$ for
\begin{equation}
  \epsilon = \frac {p^2 - (m_1 + m_2 + m_3)^2}{(m_1 + m_2 + m_3)(m_1^{-1} + m_2^{-1} + m_3^{-1})}.
\end{equation}
As long as none of the masses vanish $p^2 \neq 0$ and $\alpha_e^* \neq 0$ and therefore the Hessian matrix is non-singular.

\subsection{Sunrise with a massless propagator}
\label{sec:sunrise-one-massless}

In this section we study the sunrise integral with one massless propagator.  Since this integral is divergent we consider it in $D$ dimensions.  In this case we have
\begin{equation}
F(q_1, q_2, \alpha_1, \alpha_2) = \alpha_1 (q_1^2 - m_1^2) + \alpha_2 q_2^2 + \alpha_3 (q_3^2 - m_3^2),
\end{equation}
where $q_3 = p - q_1 - q_2$ and $\alpha_3 = 1 - \alpha_1 - \alpha_2$.

Setting $m_2 = 0$ in eqs.~\eqref{eq:alpha-relations} we find
\begin{gather}
  \frac {\alpha_1^{-1}}{\alpha_1^{-1} + \alpha_2^{-1} + (1 - \alpha_1 - \alpha_2)^{-1}} = \sqrt{\frac {q_1^2}{p^2}} = \sqrt{\frac{\epsilon + m_1^2}{p^2}}, \\
  \frac {\alpha_2^{-1}}{\alpha_1^{-1} + \alpha_2^{-1} + (1 - \alpha_1 - \alpha_2)^{-1}} = \sqrt{\frac {q_2^2}{p^2}} = \sqrt{\frac{\epsilon}{p^2}}, \\
  \frac {(1 - \alpha_1 - \alpha_2)^{-1}}{\alpha_1^{-1} + \alpha_2^{-1} + (1 - \alpha_1 - \alpha_2)^{-1}} = \sqrt{\frac {q_3^2}{p^2}} = \sqrt{\frac{\epsilon + m_3^2}{p^2}}.
\end{gather}
Adding them up we find
\begin{equation}
  p^2 = (\sqrt{\epsilon + m_1^2} + \sqrt{\epsilon} + \sqrt{\epsilon + m_3^2})^2 =
  (m_1 + m_3)^2 + 2 (m_1 + m_3) \sqrt{\epsilon} + \mathcal{O}(\epsilon).
\end{equation}

Then, it follows that
\begin{gather}
  \alpha_1^* = \frac {(\epsilon + m_1^2)^{-\frac 1 2}}{(\epsilon + m_1^2)^{-\frac 1 2} + \epsilon^{-\frac 1 2} + (\epsilon + m_3^2)^{-\frac 1 2}} = \frac {\sqrt{\epsilon}}{m_1} + \mathcal{O}(\epsilon), \\
  \alpha_2^* = \frac {\epsilon^{-\frac 1 2}}{(\epsilon + m_1^2)^{-\frac 1 2} + \epsilon^{-\frac 1 2} + (\epsilon + m_3^2)^{-\frac 1 2}} = 1 - \frac {m_1 + m_3}{m_1 m_3} \sqrt{\epsilon} + \mathcal{O}(\epsilon), \\
  \alpha_3^* = \frac {(\epsilon + m_3^2)^{-\frac 1 2}}{(\epsilon + m_1^2)^{-\frac 1 2} + \epsilon^{-\frac 1 2} + (\epsilon + m_3^2)^{-\frac 1 2}} = \frac {\sqrt{\epsilon}}{m_3} + \mathcal{O}(\epsilon),
\end{gather}
and
\begin{gather}
  \label{eq:solutions_q_1massless}
  q_1^* = \frac {m_1 p}{m_1 + m_3} - \frac {m_1 p \sqrt{\epsilon}}{(m_1 + m_3)^2} + \mathcal{O}(\epsilon), \\
  q_2^* = \frac {p \sqrt{\epsilon}}{m_1 + m_3} + \mathcal{O}(\epsilon), \\
  q_3^* = \frac {m_3 p}{m_1 + m_3} - \frac {m_3 p \sqrt{\epsilon}}{(m_1 + m_3)^2} + \mathcal{O}(\epsilon).
\end{gather}

The value of $F$ at the critical point is
\begin{equation}
F^* = \alpha_1^* ((q_1^*)^2 - m_1^2) + \alpha_2^* (q_2^*)^2 + \alpha_3^* ((q_3^*)^2 - m_3^2) = (\alpha_1^* + \alpha_2^* + \alpha_3^*) \epsilon = \epsilon.
\end{equation}

For the Hessian matrix calculation we have
\begin{equation}
  \alpha_1^* \alpha_2^* + \alpha_1^* \alpha_3^* + \alpha_2^* \alpha_3^* = (m_1^{-1} + m_3^{-1}) \sqrt{\epsilon} + \mathcal{O}(\epsilon).
\end{equation}

In conclusion, we have
\begin{equation}
  \det H^* = 2^{2 D + 2} (p^2)^2 \frac {(m_1 + m_3)^{D - 5}}{m_1^{D - 3} m_3^{D - 3}} \epsilon^{\frac {D - 1} 2} + \cdots.
\end{equation}
Since
\begin{equation}
  \epsilon = \Bigl(\frac {p^2 - (m_1 + m_3)^2}{2 (m_1 + m_3)}\Bigr)^2
\end{equation}
we have that $\det H^* \sim (p^2 - (m_1 + m_3)^2)^{D - 1}$ and since it occurs as $(\det H^*)^{-\frac 1 2}$ in the asymptotic formula, it contributes a $-\frac {D - 1} 2$ shift to the Landau exponent of the sunrise singularity.  Interestingly, in even dimensions this shift can change the nature of a singularity from square root to logarithmic and vice versa.

In this case, the uncorrected Landau exponent is $\gamma = \frac 1 2 (2 D - 4)$, which is obtained by using $\gamma = \frac 1 2 (\ell D - n - 1)$ and $\ell = 2$ (two loops), $n = 3$ (three propagators).  To this we apply the correction $-\frac {D - 1} 2$ from the behavior of the Hessian to obtain
\begin{equation}
  \label{eq:sunrise-one-massless-deformed}
  \gamma_{\text{final}} = \frac {D - 3} 2.
\end{equation}

\subsection{Sunrise with two massless propagators}
\label{sec:sunrise-two-massless}

If we have $m_2 = m_3 = 0$, then
\begin{equation}
  p^2 = (\sqrt{\epsilon + m_1^2} + 2 \sqrt{\epsilon})^2,
\end{equation}
which implies
\begin{equation}
  \epsilon = \Bigl(\frac {p^2 - m_1^2}{4 m_1}\Bigr)^2.
\end{equation}

Then, we have
\begin{gather}
  \alpha_1^* = \frac {(m_1^2 + \epsilon)^{-\frac 1 2}}{(m_1^2 + \epsilon)^{-\frac 1 2} + 2 \epsilon^{-\frac 1 2}} = \frac 1 {2 m_1} \sqrt{\epsilon} + \mathcal{O}(\epsilon), \\
  \alpha_2^* = \alpha_3^* = \frac {\epsilon^{-\frac 1 2}}{(m_1^2 + \epsilon)^{-\frac 1 2} + 2 \epsilon^{-\frac 1 2}} = \frac 1 2 - \frac 1 {4 m_1} \sqrt{\epsilon} + \mathcal{O}(\epsilon).
\end{gather}
Plugging these expansions in the expression for the determinant of the Hessian matrix in Eq.~\eqref{eq:sunrise-hessian} we find $\det H^* \sim \epsilon$.  The only contribution to the singularity arises from $(\alpha_1^*)^2 \sim \epsilon$.  Ultimately, this produces a shift of $-\frac 1 2$ to the Landau exponent of the sunrise leading singularity.  This time, the shift is dimension independent and it always changes the nature of the singularity.

\subsection{Sunrise with three massless propagators}
\label{sec:sunrise-three-massless}

If $m_1 = m_2 = m_3 = 0$, then
\begin{equation}
  p^2 = (3 \sqrt{\epsilon})^2 = 9 \epsilon.
\end{equation}
We also have $\alpha_e^* = \frac 1 3$.  Then, the only part of the determinant of the Hessian matrix which becomes singular is the $(p^2)^2$ term.  Since the Hessian matrix determinant occurs as $(\det H^*)^{-\frac 1 2}$ it contributes $(p^2)^{-1} \sim \epsilon^{-1}$.  In other words, for the fully massless case we obtain a shift of $-1$ to the Landau exponent of the sunrise leading singularity.

\subsection{Sunrise at \texorpdfstring{$p^2$ = 0}{p2=0}}
\label{sec:sunrise-zero-mom-squared}

There is another special point of the sunrise, at $p^2 = 0$.  Then, with the notations
\begin{equation}
\Delta = m_1^4 + m_2^4 + m_3^4 - 2 m_1^2 m_2^2 - 2 m_1^2 m_3^2 - 2 m_2^2 m_3^2
\end{equation}
and
\begin{gather}
(1 - z) (1 - \bar{z}) = \frac {m_2^2}{m_1^2}, \qquad
z \bar{z} = \frac {m_3^2}{m_1^2},
\end{gather}
then the answer is
\begin{equation}
\frac {D(z)}{z - \bar{z}},
\end{equation}
where $D(z)$ is the Bloch-Wigner dilogarithm
\begin{equation}
    2 i D(z) = \operatorname{Li}_2(z) - \operatorname{Li}_2(\bar{z}) + \frac 1 2 \log (z \bar{z}) \log \frac {1 - z}{1 - \bar{z}}.
\end{equation}

For the symbol we find
\begin{equation}
    \mathcal{S}(4 i D(z)) =
    \Bigl[z \bar{z} \Bigm| \frac{1 - z}{1 - \bar{z}}\Bigr]
    -\Bigl[(1 - z)(1 - \bar{z}) \Bigm| \frac {z}{\bar{z}}\Bigr],
\end{equation}
which implies that the symbol for the sunrise at $p^2 = 0$ (including the prefactor) is
\begin{multline}
    \mathcal{S}_{\text{sun}} = \frac 1 {\sqrt{\Delta}} \Biggl\{
    \Bigl[m_1^2 \Bigm| \frac {-m_1^2 + m_2^2 + m_3^2 - \sqrt{\Delta}}{-m_1^2 + m_2^2 + m_3^2 + \sqrt{\Delta}}\Bigr] +\\
    \Bigl[m_2^2 \Bigm| \frac {m_1^2 - m_2^2 + m_3^2 - \sqrt{\Delta}}{m_1^2 - m_2^2 + m_3^2 + \sqrt{\Delta}}\Bigr] +
    \Bigl[m_3^2 \Bigm| \frac {m_1^2 + m_2^2 - m_3^2 - \sqrt{\Delta}}{m_1^2 + m_2^2 - m_3^2 + \sqrt{\Delta}}\Bigr]\Biggr\}.
\end{multline}

Incidentally, it is clear on this form that for massive particles in some cases the correct first entry condition is that it be equal to the masses.  This is not universal; indeed, the massive bubble example is a counter-example.  In that case $m = 0$ is still the logarithmic singularity in the first entry (in even dimensions), but one has to accommodate the square root singularities as well.

The singularities arise at
\begin{equation}
p^2 = (m_1 \pm_{1} m_2 \pm_{2} m_3)^2 = 0
\end{equation}
which, given $p^2 = 0$ implies $m_1 \pm_{1} m_2 \pm_{2} m_3 = 0$ and at $m_e = 0$.

The leading singularity can be computed and is proportional to $\frac 1 {\sqrt{\Delta}}$, where
\begin{equation}
    \Delta = -(m_1 + m_2 + m_3)(-m_1 + m_2 + m_3)(m_1 - m_2 + m_3)(m_1 + m_2 - m_3).
\end{equation}

An alternative way to write the answer is
\begin{multline}
  \label{eq:sunrise-p2-eq-zero}
\mathcal{S}_{\text{sun}} = \frac {4 i}{\sqrt{\Delta}} \Biggl\lbrace
\Bigl[m_1 \Bigm| \mfrac {\sqrt{m_{+++}} \sqrt{m_{-++}} - i \sqrt{m_{+-+}} \sqrt{m_{++-}}}{\sqrt{m_{+++}} \sqrt{m_{-++}} + i \sqrt{m_{+-+}} \sqrt{m_{++-}}}\Bigr] \\ +
\Bigl[m_2 \Bigm| \mfrac {\sqrt{m_{+++}} \sqrt{m_{+-+}} - i \sqrt{m_{-++}} \sqrt{m_{++-}}}{\sqrt{m_{+++}} \sqrt{m_{+-+}} + i \sqrt{m_{-++}} \sqrt{m_{++-}}}\Bigr] +
\Bigl[m_3 \Bigm| \mfrac {\sqrt{m_{+++}} \sqrt{m_{++-}} - i \sqrt{m_{-++}} \sqrt{m_{+-+}}}{\sqrt{m_{+++}} \sqrt{m_{++-}} + i \sqrt{m_{-++}} \sqrt{m_{+-+}}}\Bigr]
\Biggr\rbrace.
\end{multline}

\begin{figure}
  \centering
  \includegraphics{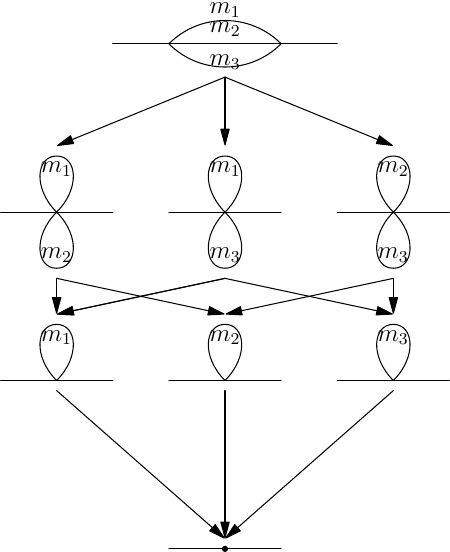}
  \caption{Hierarchy of singularities for a sunrise integral.  The dot at the bottom represents the sunrise integral while the rest of the diagrams display the cut propagators.  Not all contractions are drawn, for example there are contractions from the sunrise Landau diagram to the tadpoles and to the elementary graph as well.}
  \label{fig:sunrise-pham-web}
\end{figure}

Let us point out a few remarkable features of Eq.~\eqref{eq:sunrise-p2-eq-zero}.  If we denote $m_{\alpha \beta \gamma} = \sqrt{\alpha m_1 + \beta m_2 + \gamma m_3}$, then under $m_1 \to e^{i \pi} m_1 = -m_1$ (which corresponds to $m_1^2 \to e^{2 \pi i} m_1^2$) we have
\begin{gather}
  m_{{+}{+}{+}} \to m_{{-}{+}{+}}, \\
  m_{{-}{+}{+}} \to m_{{+}{+}{+}}, \\
  m_{{+}{-}{+}} \to m_{{-}{-}{+}} = i m_{{+}{+}{-}}, \\
  m_{{+}{+}{-}} \to m_{{-}{+}{-}} = i m_{{+}{-}{+}}.
\end{gather}
Under these transformations the prefactor transforms as
\[
  m_{{+}{+}{-}} m_{{-}{+}{+}} m_{{+}{-}{+}} m_{{+}{+}{-}} \to
  % m_{{-}{+}{+}} m_{{+}{+}{+}} i m_{{+}{+}{-}} i m_{{+}{-}{+}} =
  -m_{{+}{+}{-}} m_{{-}{+}{+}} m_{{+}{-}{+}} m_{{+}{+}{-}}.
\]
The symbol letter following $m_1$ transforms as
\[
  \frac {m_{{+}{+}{+}} m_{{-}{+}{+}} - i m_{{+}{-}{+}} m_{{+}{+}{-}}} {m_{{+}{+}{+}} m_{{-}{+}{+}} + i m_{{+}{-}{+}} m_{{+}{+}{-}}} \to
  % \frac {m_{{-}{+}{+}} m_{{+}{+}{+}} - i i m_{{+}{+}{-}} i m_{{+}{-}{+}}} {m_{{-}{+}{+}} m_{{+}{+}{+}} + i i m_{{+}{+}{-}} i m_{{+}{-}{+}}} =
  \Bigl(\frac {m_{{+}{+}{+}} m_{{-}{+}{+}} - i m_{{+}{-}{+}} m_{{+}{+}{-}}} {m_{{+}{+}{+}} m_{{-}{+}{+}} + i m_{{+}{-}{+}} m_{{+}{+}{-}}}\Bigr)^{-1}.
\]
The symbol letter following $m_2$ transforms as
\[
  \frac {m_{{+}{+}{+}} m_{{+}{-}{+}} - i m_{{-}{+}{+}} m_{{+}{+}{-}}} {m_{{+}{+}{+}} m_{{+}{-}{+}} + i m_{{-}{+}{+}} m_{{+}{+}{-}}} \to
  % \frac {m_{{-}{+}{+}} i m_{{+}{+}{-}} - i m_{{+}{+}{+}} i m_{{+}{-}{+}}} {m_{{-}{+}{+}} i m_{{+}{+}{-}} + i m_{{+}{+}{+}} i m_{{+}{-}{+}}} =
  \Bigl(\frac {m_{{+}{+}{+}} m_{{+}{-}{+}} - i m_{{-}{+}{+}} m_{{+}{+}{-}}} {m_{{+}{+}{+}} m_{{+}{-}{+}} + i m_{{-}{+}{+}} m_{{+}{+}{-}}}\Bigr)^{-1}.
\]
Finally, the symbol letter following $m_3$ transforms as
\[
  \frac {m_{{+}{+}{+}} m_{{+}{+}{-}} - i m_{{-}{+}{+}} m_{{+}{-}{+}}}{m_{{+}{+}{+}} m_{{+}{+}{-}} + i m_{{-}{+}{+}} m_{{+}{-}{+}}} \to
  % \frac {m_{{-}{+}{+}} i m_{{+}{-}{+}} - i m_{{+}{+}{+}} i m_{{+}{+}{-}}} {m_{{-}{+}{+}} i m_{{+}{-}{+}} + i m_{{+}{+}{+}} i m_{{+}{+}{-}}} =
  \Bigl(\frac {m_{{+}{+}{+}} m_{{+}{+}{-}} - i m_{{-}{+}{+}} m_{{+}{-}{+}}}{m_{{+}{+}{+}} m_{{+}{+}{-}} + i m_{{-}{+}{+}} m_{{+}{-}{+}}}\Bigr)^{-1}.
\]

The multiplicative sign from the prefactor cancels the sign from inverting the symbol letters.  Hence, the only term which is not invariant under $m_1^2 \to e^{2 \pi i} m_1^2$ is the first letter $m_1$ in the first term of Eq.~\eqref{eq:sunrise-p2-eq-zero}.

After taking the $m_1^2 \to e^{2 \pi i} m_1^2$ monodromy we obtain the discontinuity
\begin{equation}
  \label{eq:sunrise-p2-eq-zero-tadpole1}
  \frac {4 i (-2 \pi i)}{\sqrt{\Delta}}
\Bigl[ \frac {\sqrt{m_{+++}} \sqrt{m_{-++}} - i \sqrt{m_{+-+}} \sqrt{m_{++-}}}{\sqrt{m_{+++}} \sqrt{m_{-++}} + i \sqrt{m_{+-+}} \sqrt{m_{++-}}}\Bigr].
\end{equation}
These $m_e^2 = 0$ singularities correspond to the contraction of tadpole graphs to the elementary graph in fig.~\ref{fig:sunrise-pham-web}.

This has a logarithmic singularity at
\begin{multline}
    0 = (\sqrt{m_{+++}} \sqrt{m_{-++}} - i \sqrt{m_{+-+}} \sqrt{m_{++-}})\times \\ (\sqrt{m_{+++}} \sqrt{m_{-++}} + i \sqrt{m_{+-+}} \sqrt{m_{++-}}) = 4 m_2^2 m_3^2.
\end{multline}
These two possibilities correspond to the two contractions of the double tadpole to a single tadpole diagram.

Therefore, taking a discontinuity around $m_1^2 \to e^{2 \pi i} m_1^2$, followed by a discontinuity around $m_2^2 \to e^{2 \pi i} m_2^2$ yields the leading singularity
\[
  \frac {4 i (-2 \pi i)^2}{\sqrt{m_{+++}} \sqrt{m_{-++}} \sqrt{m_{+-+}} \sqrt{m_{++-}}}.
\]
This then has all the singularities of the sunrise Landau diagram, including for negative values of $\alpha$.

We could also think of the sunrise singularity of the tadpole absorption integral in Eq.~\eqref{eq:sunrise-p2-eq-zero-tadpole1}.  The sunrise singularity is of algebraic (not logarithmic) type but this time the only negative $\alpha$ which are allowed are the ones on the already cut propagators.  Indeed, by the same mechanism which operates for the bubble integral and differentiates between the threshold and pseudo-threshold, we can arrange so that the singularities $m_1 + m_2 + m_3 = 0$ (with all $\alpha > 0$) and $-m_1 + m_2 + m_3 = 0$ (with $\alpha_1 < 0$ and $\alpha_2, \alpha_3 > 0$) occur on the main sheet for the one-cut integral, while $m_1 - m_2 + m_3 = 0$ (which corresponds to $\alpha_2 < 0$ and $\alpha_1, \alpha_3 > 0$) and $m_1 + m_2 - m_3 = 0$ (which corresponds to $\alpha_3 < 0$ and $\alpha_1, \alpha_2 > 0$) are relegated to the second sheet and only become visible after further analytic continuation.  It is essential for the Galois symmetry argument that the two allowed possibilities $\sqrt{m_{+++}}$ and $\sqrt{m_{-++}}$ occur together.

Let us solve the Landau equations for the sunrise integral with external momentum $p$ such that $p^2 = 0$.\footnote{The previous approach fails since the Hessian determinant is proportional to $p^2$ which now vanishes.}  We have
\begin{equation}
  F(\alpha_1, \alpha_2, q_1, q_2) = \sum_{e = 1}^3 \alpha_e (q_e^2 - m_e^2)
\end{equation}
where $q_3 = p - q_1 - q_2$ and $\alpha_3 = 1 - \alpha_1 - \alpha_2$.

We obtain the following conditions
\begin{gather}
  q_1^2 - m_1^2 = q_2^2 - m_2^2 = q_3^2 - m_3^2 = \epsilon, \\
  \alpha_1 q_1 = \alpha_2 q_2 = \alpha_3 q_3 = v.
\end{gather}

Previously we derived
\begin{equation}
  q_e = \frac {p \alpha_e^{-1}}{\sum_{e'} \alpha_{e'}^{-1}}.
\end{equation}
This will not work anymore if we take $p^2 = 0$ since then $q_e^2 = 0$, which can not be on-shell if the masses are non-vanishing, as we assume.

Note that we have assumed that we can divide by $\sum_e \alpha_e^{-1}$ which therefore has to be non-vanishing.  If however $\sum_e \alpha_e^{-1} = 0$ we have that $p = 0$ (and therefore $p^2 = 0$), while the on-shell conditions for $q_e$ can be also satisfied.

The identity $\sum_e \alpha_e^{-1} = 0$ can not be satisfied when all $\alpha_e > 0$ so it does not arise on the first sheet.  Let us assume that we are looking at the one-cut integral where we have cut the propagator with momentum $q_3$.  If we are now studying the singularities of this cut integral, we have $\alpha_1, \alpha_2 > 0$ while the sign of $\alpha_3$ is a priori unconstrained.  We need to take $\alpha_3 < 0$ in order to be able to satisfy the identity $\sum_e \alpha_e^{-1} = 0$.

Squaring the Landau loop equations we obtain
\begin{equation}
  \alpha_1^2 (m_1^2 + \epsilon) = \alpha_2^2 (m_2^2 + \epsilon) = \alpha_3^2 (m_3^2 + \epsilon).
\end{equation}
Keeping in mind the signs of $\alpha$, after some manipulations we find the condition
\begin{equation}
  \sqrt{m_1^2 + \epsilon} + \sqrt{m_2^2 + \epsilon} - \sqrt{m_3^2 + \epsilon} = 0
\end{equation}
When this is satisfied we have $\Delta = 4 (m_1^2 + m_1 m_2 + m_2^2) \epsilon + \mathcal{O}(\epsilon^2)$.

One solution for the $q_e$ at the location of the Landau locus is
\begin{gather}
  q_1^0 = \sqrt{m_1^2 + \epsilon}, \qquad \vec{q}_1 = \vec{0}, \\
  q_2^0 = \sqrt{m_2^2 + \epsilon}, \qquad \vec{q}_2 = \vec{0}, \\
  q_3^0 = -\sqrt{m_3^2 + \epsilon}, \qquad \vec{q}_3 = \vec{0}.
\end{gather}
This solution is not unique and every Lorentz transformation of this particular solution is also a solution.

The Hessian matrix is
\begin{equation}
  H = 2 \begin{pmatrix}
    (1 - \alpha_2) \eta & (1 - \alpha_1 - \alpha_2) \eta & -q_2 & -(q_1 + q_2) \\
    (1 - \alpha_1 - \alpha_2) \eta & (1 - \alpha_1) \eta & -(q_1 + q_2) & -q_1 \\
    -q_2 & -(q_1 + q_2) & 0 & 0 \\
    -(q_1 + q_2) & -q_2 & 0 & 0
  \end{pmatrix} = 2 \begin{pmatrix}
    A & B \\
    B^t & 0
  \end{pmatrix},
\end{equation}
where
\begin{gather}
  A = \begin{pmatrix}
    1 - \alpha_2 & 1 - \alpha_1 - \alpha_2 \\
    1 - \alpha_1 - \alpha_2 & 1 - \alpha_1
  \end{pmatrix} \otimes \eta, \\
  B = -\begin{pmatrix}
    q_2 & q_1 + q_2 \\
    q_1 + q_2 & q_1
  \end{pmatrix}
\end{gather}

Before we computed this using the inverse of $A$, but now $A$ is not invertible.  Indeed, since $\alpha_1 \alpha_2 + \alpha_1 \alpha_3 + \alpha_2 \alpha_3 = 0$ we have $\det A = 0$.  Since the $2 \times 2$ matrix
\begin{equation}
  \begin{pmatrix}
    1 - \alpha_2 & 1 - \alpha_1 - \alpha_2 \\
    1 - \alpha_1 - \alpha_2 & 1 - \alpha_1
  \end{pmatrix} =
  \begin{pmatrix}
    \alpha_1 + \alpha_3 & \alpha_3 \\
    \alpha_3 & \alpha_2 + \alpha_3
  \end{pmatrix}
\end{equation}
has an eigenvector $\left(\begin{smallmatrix} \alpha_2 \\ \alpha_1\end{smallmatrix}\right)$ with eigenvalue zero.

The question becomes what to do in such a degenerate case where the determinant of the Hessian matrix vanishes.  A first guess would be to remove the zero modes and take the determinant of that.  However, here is an example which shows that one needs to go beyond quadratic order to understand the asymptotic behavior of the integral.  Consider
\begin{equation}
\int_{\mathbb{R}^2}  \frac {d x d y}{\epsilon + x^2 + y^4},
\end{equation}
and try to find the behavior in the limit $\epsilon \to 0$.

This is actually very easy.  We make a change of variable\footnote{This is similar to how the massless bubble integral was treated.  In fact, the change of variables below can be thought of as a weighted blow-up of the singular curve $x^2 + y^4 = 0$ at $(x, y) = (0, 0)$.}
\begin{equation}
x = \sqrt{\epsilon} u, \qquad
y = \epsilon^{\frac 1 4} v.
\end{equation}
Using this we can determine the complete dependence on $\epsilon$ (not just when $\epsilon \to 0$).  The integral becomes
\begin{equation}
\epsilon^{-\frac 1 4} \int_{\mathbb{R}^2} \frac {d u d v}{1 + u^2 + v^4} = \frac \pi 2 B(\frac 1 4, \frac 1 4) \epsilon^{-\frac 1 4}.
\end{equation}

A naive calculation where we look at the degenerate Hessian
\begin{equation}
\begin{pmatrix}
2 & 0 \\
0 & 0
\end{pmatrix}
\end{equation}
at the critical point $(x, y) = (0, 0)$ and we drop the zero modes will not be able to reproduce the exponent of $\frac 1 4$.  When the Hessian is degenerate one needs to go to higher and the analysis acquires a tropical geometry flavor as we describe in more detail below.

The example above is a toy model.  To describe in more details what happens in a physical case we consider an integral $\int \frac {d^p x d^q y}{F^r}$, where $p, q$ are some integers, $r$ is an integer of half an odd integer (in odd dimensions) and the expansion of $F$ around the critical point $(x, y) = (x^*, y^*)$ reads
\begin{multline}
  F(x, y) = \epsilon +
  \frac 1 2 \sum_{i, j = 1}^p F_{i j} (x_i - x_i^*) (x_j - x_j^*) + \\
  \frac 1 {3!} \Bigl(
  \sum_{i, j, k = 1}^p F_{i j k} (x_i - x_i^*) (x_j - x_j^*) (x_k - x_k^*) + \\
  3 \sum_{i j = 1}^p \sum_{a = 1}^q F_{i j a} (x_i - x_i^*) (x_j - x_j^*) (y_a - y_a^*) + \\
  3 \sum_{i = 1}^p \sum_{a, b = 1}^q F_{i a b} (x_i - x_i^*) (y_a - y_a^*) (y_b - y_b^*) + \\
  \sum_{a, b, c = 1}^q F_{a b c} (y_a - y_a^*) (y_b - y_b^*) (y_c - y_c^*)
  \Bigr)
\end{multline}
We have stopped at the third order since this is the highest order arising in a Feynman parametrized integral; the momenta are second order and the $\alpha$ are first order.

In the region which contributes to the leading order in $\epsilon$ we take $x_i - x_i^* = \mathcal{O}(\epsilon^{\frac 1 2})$.  In this case, the terms $F_{i j k} (x_i - x_i^*) (x_j - x_j^*) (x_k - x_k^*)$ are always subleading.

If $F_{a b c}$ is non-vanishing then we need to take $y_a - y_a^* = \mathcal{O}(\epsilon^{\frac 1 3})$ in order to obtain a sizeable contribution to the integral.  If $F_{a b c} = 0$ for all $a, b, c = 1, \dotsc, q$, then take $y_a - y_a^* = \mathcal{O}(\epsilon^{\frac 1 4})$ so that the terms $F_{i a b} (x_i - x_i^*) (y_a - y_a^*) (y_b - y_b^*) = \mathcal{O}(\epsilon)$.  More complicated options are possible.  For example, a subset of $F_{a b c}$ could vanish in which case we take $y - y^* = \mathcal{O}(\epsilon^{\frac 1 3})$ for part of the $y$'s and $y - y^* = \mathcal{O}(\epsilon^{\frac 1 4})$ for the remaining ones.\footnote{This is not in contradiction with the analysis in ref.~\cite{Bourjaily:2022vti} since the vanishing of the Hessian determinant amounts to a fine-tuning the external kinematics.}  A general analysis becomes unwieldy and it is best to apply this idea on explicit examples.

For the sunrise example for the third order expansion we need
\begin{gather}
  \frac {\partial^3 F}{\partial \alpha_1 \partial q_1^2} = 0, \qquad
  \frac {\partial^3 F}{\partial \alpha_1 \partial q_1 \partial q_2} = -2 \eta, \qquad
  \frac {\partial^3 F}{\partial \alpha_1 \partial q_2^2} = -2 \eta, \\
  \frac {\partial^3 F}{\partial \alpha_2 \partial q_1^2} = -2 \eta, \qquad
  \frac {\partial^3 F}{\partial \alpha_2 \partial q_1 \partial q_2} = -2 \eta, \qquad
  \frac {\partial^3 F}{\partial \alpha_2 \partial q_2^2} = 0.
\end{gather}

\subsection{Sunrise at the pseudo-threshold}
\label{sec:sunrise-pseudo-threshold}

The sunrise with masses $m_1, m_2, m_3$ evaluated at the pseudo-threshold $M_i = m_1 + m_2 + m_3 - 2 m_i$ is polylogarithmic (see ref.~\cite[Eq.~3.13]{Bloch:2013tra}).  Up to some numerical factors it reads
\begin{equation}
\frac 1 {\sqrt{m_1 m_2 m_3 M_i}} \sum_{j = 1}^3 \partial_{m_j} M_i \widetilde{D}\Biggl(\sqrt{\frac {m_j^2 M_i}{m_1 m_2 m_3}}\Biggr),
\end{equation}
where
\begin{equation}
2 i \tilde{D}(z) = \operatorname{Li}_2(z) - \operatorname{Li}_2(-z) + \frac 1 2 \log(z^2) \log \frac {1 - z}{1 + z}.
\end{equation}

It has symbol
\begin{equation}
\mathcal{S}(2 i \tilde{D}(z)) = [z \mid 1 -  z] - [z \mid 1 + z].
\end{equation}
The symbol of the sunrise is
\begin{multline}
\label{eq:sunrise-pseudothreshold-symbol}
\frac 1 {\sqrt{m_1} \sqrt{m_2} \sqrt{m_3} \sqrt{M_3}} \Bigl(
\Bigl[m_1 \Bigm| \frac {\sqrt{m_2} \sqrt{m_3} - \sqrt{m_1} \sqrt{M_3}} {\sqrt{m_2} \sqrt{m_3} + \sqrt{m_1} \sqrt{M_3}}\Bigr] + \\
\Bigl[m_2 
\Bigm| \frac {\sqrt{m_1} \sqrt{m_3} - \sqrt{m_2} \sqrt{M_3}}{\sqrt{m_1} \sqrt{m_3} + \sqrt{m_2} \sqrt{M_3}}\Bigr] -
\Bigl[m_3 
\Bigm| \frac {\sqrt{m_1} \sqrt{m_2} - \sqrt{m_3} \sqrt{M_3}}{\sqrt{m_1} \sqrt{m_2} + \sqrt{m_3} \sqrt{M_3}}\Bigr]\Bigr).
\end{multline}

The reason for splitting the square roots as in Eq.~\eqref{eq:sunrise-pseudothreshold-symbol} was discussed in sec.~\ref{sec:massive-bubble-2D}.

To find this formula for the symbol we have used
\begin{equation}
\label{eq:sunrise-pseudothreshold-identity}
\frac {(1 - \sqrt{a_1})(1 - \sqrt{a_2})(1 + \sqrt{a_3})}{(1 + \sqrt{a_1})(1 + \sqrt{a_2})(1 - \sqrt{a_3})} = 1,
\end{equation}
where
\begin{gather}
\sqrt{a_1} = \frac {\sqrt{M_3} \sqrt{m_1}}{\sqrt{m_2} \sqrt{m_3}}, \qquad
\sqrt{a_2} = \frac {\sqrt{M_3} \sqrt{m_2}}{\sqrt{m_1} \sqrt{m_3}}, \qquad
\sqrt{a_3} = \frac {\sqrt{M_3} \sqrt{m_3}}{\sqrt{m_1} \sqrt{m_2}}.
\end{gather}

To prove Eq.~\eqref{eq:sunrise-pseudothreshold-identity} we use
\begin{gather}
\sqrt{a_1} \sqrt{a_2} = \frac {M_3}{m_3}, \qquad
\sqrt{a_1} \sqrt{a_3} = \frac {M_3}{m_2}, \qquad
\sqrt{a_2} \sqrt{a_3} = \frac {M_3}{m_1}.
\end{gather}
Extracting the square roots (assuming all masses, including $M_3$, are positive) we find
\begin{gather}
\frac 1 {\sqrt{a_1} \sqrt{a_2}} = \frac {m_3}{M_3}, \qquad
\frac 1 {\sqrt{a_1} \sqrt{a_3}} = \frac {m_2}{M_3}, \qquad
\frac 1 {\sqrt{a_2} \sqrt{a_3}} = \frac {m_1}{M_3}.
\end{gather}

Next, using $M_3 = m_1 + m_2 - m_3$, we find
\begin{equation}
    \frac 1 {\sqrt{a_2} \sqrt{a_3}} + \frac 1 {\sqrt{a_1} \sqrt{a_3}} - \frac 1 {\sqrt{a_1} \sqrt{a_2}} = 1,
\end{equation}
which implies that
\begin{equation}
    \sqrt{a_1} + \sqrt{a_2} - \sqrt{a_3} = \sqrt{a_1} \sqrt{a_2} \sqrt{a_3}.
\end{equation}
Upon expanding the numerator and the denominator of Eq.~\eqref{eq:sunrise-pseudothreshold-identity} and subtracting them, we find the same identity.

Thanks to the Galois symmetry, the square roots in the symbol of Eq.~\eqref{eq:sunrise-pseudothreshold-symbol} are actually not singularities,\footnote{The situation is actually a bit more complicated, see below.} until we take the logarithmic monodromies which remove the first and second symbol entries and leave only the leading singularity $\frac 1 {\sqrt{m_1} \sqrt{m_2} \sqrt{m_3} \sqrt{M_3}}$.

Curiously, the symbol of the sunrise at pseudo-threshold is best written in a form which is \emph{not} manifestly dimensionless.  Indeed, we have, schematically
\begin{equation}
  [m_1 | B_1] + [m_2 | B_2] + [m_3 | B_3],
\end{equation}
with $B_1 B_2 B_3 = 1$. This can be written as
\begin{equation}
[m_1/m_3| B_1] + [m_2/m_3 | B_2],
\end{equation}
at the cost of breaking the symmetry.

\subsubsection*{Prefactor}
\label{sec:sunrise-pseudo-threshold-prefactor}

We want to compute the leading singularity which is an one-dimensional integral left over after cutting all three propagators.  It reads
\begin{equation}
  (-2 \pi i)^3 \int d^2 q_1 d^2 q_2 \delta(q_1^2 - m_1^2) \delta(q_2^2 - m_2^2) \delta(q_3^2 - m_3^2).
\end{equation}
In terms of Leray residues the integral reads (we will drop the factors of $\pi$ in the following)
\begin{equation}
  \int \frac {d^2 q_1 \wedge d^2 q_2}{d (q_1^2 - m_1^2) \wedge d (q_2^2 - m_2^2) \wedge d (q_3^2 - m_3^2)}
\end{equation}
The factor in the denominator reads
\begin{multline}
  d (q_1^2 - m_1^2) \wedge d (q_2^2 - m_2^2) \wedge d (q_3^2 - m_3^2) = \\
  8 {q_1}_{\mu} {q_2}_{\nu} (q_1 + q_2 - p)_{\rho} \Bigl(
  -\epsilon^{\mu \rho} d^2 q_1 \wedge d q_2^\nu +
  \epsilon^{\nu \rho} d q_1^\mu \wedge d^2 q_2\Bigr),
\end{multline}
where we have used $q_3 = p - q_1 - q_2$.  If we wedge this quantity with $\xi \cdot d q_1$, where $\xi$ is a vector, we obtain $8 \epsilon(q_1, \xi) \epsilon(q_2, q_1 - p)$, where we have used $d q_1^\mu d q_1^\nu = \epsilon^{\mu \nu} d^2 q_1$.

We then obtain
\begin{equation}
  \frac {d^2 q_1 \wedge d^2 q_2}{d (q_1^2 - m_1^2) \wedge d (q_2^2 - m_2^2) \wedge d (q_3^2 - m_3^2)} =
  -\frac {\xi \cdot d q_1}{8 \epsilon(q_1, \xi) \epsilon(q_2, q_1 - p)},
\end{equation}
where both expressions are evaluated on the support of $q_e^2 = m_e^2$ for $e = 1, 2, 3$.  In the following we will take $\xi = p$.

We can now use a convenient dual space picture where $q_1 = v_1 - v_0$, $q_2 = v_2 - v_1$, $q_3 = v_3 - v_2$ and $p = q_1 + q_2 + q_3 = v_3 - v_0$.  Since we have three vectors $v_1 - v_0$, $v_2 - v_0$ and $v_3 - v_0$ in a two-dimensional space we have one Gram determinant constraint
\begin{equation}
  \det ((v_i - v_0) \cdot (v_j - v_0))_{i, j = 1, 2, 3} = 0.
\end{equation}
Plugging in $(v_1 - v_0)^2 = m_1^2$, $(v_2 - v_1)^2 = m_2^2$, $(v_3 - v_2)^2 = m_3^2$ and $(v_3 - v_0)^2 = p^2$, together with the notations $x^2 = (v_2 - v_0)^2$ and $y^2 = (v_3 - v_1)^2$ we find
\begin{equation}
  \det
  \begin{pmatrix}
    m_1^2 & \frac 1 2 (m_1^2 + x^2 - m_2^2) & \frac 1 2 (m_1^2 + p^2 - y^2) \\
    \frac 1 2 (m_1^2 + x^2 - m_2^2) & x^2 & \frac 1 2 (x^2 + p^2 - m_3^2) \\
    \frac 1 2 (m_1^2 + p^2 - y^2) & \frac 1 2 (x^2 + p^2 - m_3^2)& p^2
  \end{pmatrix} = 0.
\end{equation}
Expanded out, this equation reads
\begin{multline}
  f(x, y) = -x^4 y^2 - x^2 y^4 +
  2 x^2 y^2 (m_1^2 + m_2^2 + m_3^2 + m_1 m_2 - m_1 m_3 - m_2 m_3) + \\
  x^2 (2 m_1 + m_2 - m_3) (m_2 + m_3) (m_2 - m_3)^2 + \\
  y^2 (2 m+3 - m_1 - m_2) (m_1 + m_2)^2 (m_1 - m_2) - \\
  2 (m_1 m_2 + m_2^2 + m_1 m_3 - m_2 m_3) (m_1 + m_2)^2 (m_2 - m_3)^2 = 0,
\end{multline}
where we have used the pseudo-threshold value $p^2 = (m_1 + m_2 - m_3)^2$.

This is the equation of an elliptic curve with a double point singularity at $x = m_1 + m_2$ and $y = m_2 - m_3$.  We make a change of variables which blows up the singularity at this point and parametrizes the points of the curve by the slope of a pencil of lines through this singular point:
\begin{gather}
  x^2 = (m_1 + m_2)^2 + u, \qquad
  y^2 = (m_2 - m_3)^2 + u v.
\end{gather}
Then the equation of the curve becomes
\begin{equation}
  u^2 \Bigl(((m_1 + m_2)^2 + u) v^2 +
  2 (m_1 m_2 + m_2^2 + m_1 m_3 - m_2 m_3 + u/2) v +
  (m_2 - m_3)^2\Bigr) = 0.
\end{equation}

Going back to the one-form we need to integrate, in terms of the new coordinates $u$ and $v$ we get that
\begin{multline}
  \epsilon(q_1, p) \epsilon(q_2, q_1 - p) = -\det \begin{pmatrix}
    q_1 \cdot q_2 & q_1 \cdot (q_1 - p) \\
    p \cdot q_2 & p \cdot (q_1 - p)
  \end{pmatrix} = \\
  -\det \begin{pmatrix}
    \frac 1 2 (x^2 - m_1^2 - m_2^2) & \frac 1 2 (y^2 + m_1^2 - p^2) \\
    \frac 1 2 (x^2 + y^2 - m_1^2 - m_3^2) & \frac 1 2 (x^2 - p^2 - m_3^2)
  \end{pmatrix} = \\
  -\frac 1 4 (2 m_1 m_2 v + 2 m_2^2 v + 2 m_1 m_3 v - 2 m_2 m_3 v + u v^2 + 2 m_2^2 - 4 m_2 m_3 + 2 m_3^2 + 2 u v) u.
\end{multline}

The differential form becomes:
\begin{equation}
  -\frac 1 4 \frac {d (u v)}{u (2 m_1 m_2 v + 2 m_2^2 v + 2 m_1 m_3 v - 2 m_2 m_3 v + u v^2 + 2 m_2^2 - 4 m_2 m_3 + 2 m_3^2 + 2 u v)}.
\end{equation}

We can solve for $u$ as a function of $v$ from the equation of the curve and eliminate $u$ from the differential form.  Then, we obtain
\begin{equation}
  \frac {d v}{(m_1 + m_2)^2 v^2 + 2 (m_1 m_2 + m_2^2 + m_1 m_3 - m_2 m_3) v + (m_2 - m_3)^2}.
\end{equation}
This can be partial-fractioned to
\begin{equation}
  \frac 1 {\sqrt{\Delta}} \Bigl(\frac {d v}{v - v_+} - \frac {d v}{v - v_-}\Bigr),
\end{equation}
where $\Delta = 16 m_1 m_2 m_3 M_3$.  Taking the residue reproduces the prefactor $\frac 1 {\sqrt{m_1 m_2 m_3 M_3}}$.

\subsubsection*{Analysis of the symbol}
\label{sec:sunrise-pseudo-threshold-analysis}

The first entries in the symbol of Eq.~\eqref{eq:sunrise-pseudothreshold-symbol} are at $m_e = 0$ and correspond to tadpole diagrams.

The logarithmic singularities in the second entry following $m_1$ are at
\begin{equation}
m_2 m_3 - M_3 m_1 = (m_3 - m_1) (m_1 + m_2) = 0.
\end{equation}
We will now show this is the subset of the singularities $p^2 = (m_1 \pm_{1} m_2 \pm_{2} m_3)^2$ for $p^2 = M_3^2$, for which $\alpha_2, \alpha_3 \geq 0$.

The general Landau analysis for a sunrise integral in this kinematics goes as follows.  We have $\alpha_1 q_1 = \alpha_2 q_2 = \alpha_3 q_3 = v$.  Assuming that none of the $\alpha$ vanish we have $q_e = \frac v {\alpha_e}$ and using momentum conservation we have
\begin{equation}
  v = p \Bigl(\alpha_1 ^{-1} + \alpha_2^{-1} + \alpha_3^{-1}\Bigr)^{-1}.
\end{equation}
Using the on-shell conditions we have
\begin{gather}
  \alpha_2 = \pm_1 \alpha_1 \frac {m_1}{m_2}, \qquad
  \alpha_3 = \pm_2 \alpha_1 \frac {m_1}{m_3}.
\end{gather}
Using $\alpha_1 + \alpha_2 + \alpha_3 = 1$ we have
\begin{equation}
  \alpha_1 = \Bigl(1 \pm_1 \frac {m_1}{m_2} \pm_2 \frac {m_1}{m_3}\Bigr)^{-1}.
\end{equation}
It follows that
\begin{equation}
  \alpha_1 ^{-1} + \alpha_2^{-1} + \alpha_3^{-1} =
  \alpha_1^{-1} \Bigl(1 \pm_1 \frac {m_2}{m_1} \pm_2 \frac {m_3}{m_1}\Bigr).
\end{equation}
Finally,
\begin{equation}
  p^2 = v^2 \Bigl(\alpha_1 ^{-1} + \alpha_2^{-1} + \alpha_3^{-1}\Bigr)^2,
\end{equation}
and using, for example, $v^2 = \alpha_1^2 m_1^2$ and the previous identity we find
\begin{equation}
  p^2 = (m_1 \pm_1 m_2 \pm_2 m_3)^2.
\end{equation}

After cutting one propagator (which corresponds to the initial entry $m_1$ in the symbol), the singularities that can arise are those whose as-yet uncut propagators have non-negative $\alpha$.  Indeed, we have
\begin{equation}
(m_1 + m_2 - m_3)^2 - (m_1 + m_2 + m_3)^2 = - 4 m_3 (m_1 + m_2).
\end{equation}

For $p^2 = (m_1 + m_2 + m_3)^2$ we have $\pm_1 = +$ and $\pm_2 = +$ and therefore
\begin{gather}
  \alpha_1 =\Bigl(1 + \frac {m_1}{m_2} + \frac {m_1}{m_3}\Bigr)^{-1}, \\
  \alpha_2 =\Bigl(1 + \frac {m_2}{m_1} + \frac {m_2}{m_3}\Bigr)^{-1}, \\
  \alpha_3 =\Bigl(1 + \frac {m_3}{m_1} + \frac {m_3}{m_2}\Bigr)^{-1}.
\end{gather}
For $m_3 = 0$ we get
\begin{gather}
  \alpha_1 = \alpha_2 = 0, \qquad
  \alpha_3 = 1,
\end{gather}
while for $m_1 + m_2 = 0$ we get
\begin{gather}
  \alpha_1 = \frac {m_3}{m_1}, \qquad
  \alpha_2 = -\frac {m_3}{m_1}, \qquad
  \alpha_3 = 1.
\end{gather}

The $m_3 = 0$ singularity is of square root type.  The branch cuts have to be placed such that the second entry in the part of the symbol in Eq.~\eqref{eq:sunrise-pseudothreshold-symbol} starting by $m_1$ yields square root singularities when $m_3 = 0$ (and also $m_2 = 0$ as we show below) but not when $m_1 = 0$ or $M_3 = 0$.

The logarithmic singularity at $m_1 + m_2 = 0$ arises when $\alpha_2 > 0$ but $\alpha_1$ is allowed to be negative.  Therefore we need to take $m_1 m_3 < 0$.

We also have
\begin{equation}
(m_1 + m_2 - m_3)^2 - (-m_1 + m_2 + m_3)^2 = 4 m_2 (m_1 - m_3),
\end{equation}
whose corresponding $\alpha$ are
\begin{gather}
  \alpha_1 = \Bigl(1 - \frac {m_1}{m_2} - \frac {m_1}{m_3}\Bigr)^{-1}, \\
  \alpha_2 = \Bigl(1 - \frac {m_2}{m_1} + \frac {m_2}{m_3}\Bigr)^{-1}, \\
  \alpha_3 = \Bigl(1 - \frac {m_3}{m_1} + \frac {m_3}{m_2}\Bigr)^{-1}.
\end{gather}
For $m_2 = 0$ we have
\begin{gather}
  \alpha_1 = \alpha_3 = 0, \qquad
  \alpha_2 = 1,
\end{gather}
while for $m_1 - m_3 = 0$ we have
\begin{gather}
  \alpha_1 = -\frac {m_2}{m_1}, \qquad
  \alpha_3 = \frac {m_2}{m_1}, \qquad
  \alpha_2 = 1.
\end{gather}

These factors $(m_1 + m_2)$ and $(m_1 - m_3)$ are precisely the ones arising in
\begin{equation}
(\sqrt{m_2 m_3} - \sqrt{M_3 m_1})(\sqrt{m_2 m_3} + \sqrt{M_3 m_1}) = (m_3 - m_1) (m_1 + m_2).
\end{equation}

Let us now analyze the Hessians that arise for this kinematics.  The solutions for the sunrise Landau equations are $p^2 = (m_1 \pm_1 m_2 \pm_2 m_3)^2$.  If we take $p^2 = (m_1 + m_2 - m_3)^2$, then we have the following possibilities:
\begin{enumerate}
  \item $\pm_1 = +$, $\pm_2 = -$,
  \item $\pm_1 = +$, $\pm_2 = +$
  \item $\pm_1 = -$, $\pm_2 = +$
  \item $\pm_1 = -$, $\pm_2 = -$
\end{enumerate}

In the second case we have
\begin{equation}
  (m_1 + m_2 + m_3)^2 = (m_1 + m_2 - m_3)^2
\end{equation}
which implies that either $m_3 = 0$ or $m_1 + m_2 = 0$.

In the case $m_3 = 0$ we have the critical value solutions
\begin{gather}
  p = (m_1 + m_2, 0), \qquad
  q_1^* = (m_1, 0), \qquad
  q_2^* = (m_2, 0), \qquad
  q_3^* = (0, 0), \\
  \alpha_1^* = 0, \qquad
  \alpha_2^* = 0, \qquad
  \frac {\alpha_1^*}{\alpha_2^*} = \frac {m_2}{m_1}, \qquad
  \alpha_3^* = 1,
\end{gather}
where we have picked the external momentum along the time direction.

% case 2, m_3 = 0
We find for the Hessian matrix
\[
  \frac H 2 = \left(
    \begin{array}{cc|cc}
      \eta & \eta & \begin{smallmatrix} m_1 \\ 0 \end{smallmatrix} & \begin{smallmatrix} 0 \\ 0 \end{smallmatrix} \\
      \eta & \eta & \begin{smallmatrix} 0 \\ 0 \end{smallmatrix} & \begin{smallmatrix} m_2 \\ 0 \end{smallmatrix} \\
      \hline
      \begin{smallmatrix} m_1 & 0 \end{smallmatrix} & \begin{smallmatrix} 0 & 0 \end{smallmatrix} & 0 & 0 \\
      \begin{smallmatrix} 0 & 0 \end{smallmatrix} & \begin{smallmatrix} m_2 & 0 \end{smallmatrix} & 0 & 0
    \end{array}
  \right)
\]
The determinant of this matrix vanishes.

If, instead, $m_1 + m_2 = 0$ we find
\begin{gather}
  p = (-m_3, 0), \qquad
  q_1^* = (m_1, 0), \qquad
  q_2^* = (-m_1, 0), \qquad
  q_3^* = (m_3, 0), \\
  \alpha_1^* = \frac {m_3}{m_1}, \qquad
  \alpha_2^* = -\frac {m_3}{m_1}, \qquad
  \alpha_3^* = 1.
\end{gather}

% case 2, m_1 + m_2 = 0
We find for the Hessian matrix
\[
  \frac H 2 = \left(
    \begin{array}{cc|cc}
      \bigl(1 + \frac {m_3}{m_1}\bigr) \eta & \eta & \begin{smallmatrix} m_1 - m_3 \\ 0 \end{smallmatrix} & \begin{smallmatrix} m_3 \\ 0 \end{smallmatrix} \\
      \eta & \bigl(1 - \frac {m_3}{m_1}\bigr) \eta & \begin{smallmatrix} m_3 \\ 0 \end{smallmatrix} & \begin{smallmatrix} -m_1 - m_3 \\ 0 \end{smallmatrix} \\
      \hline
      \begin{smallmatrix} m_1 - m_3 & 0 \end{smallmatrix} & \begin{smallmatrix} m_3 & 0 \end{smallmatrix} & 0 & 0 \\
      \begin{smallmatrix} m_3 & 0 \end{smallmatrix} & \begin{smallmatrix} -m_1 - m_3 & 0 \end{smallmatrix} & 0 & 0
    \end{array}
  \right)
\]
Also,
\[
  \det \frac H 2 = -m_1^2 m_3^2.
\]

Recalling that the prefactor of the integral is $\frac 1 {\sqrt{m_1 m_2 m_3 M_3}}$ and evaluating this at $m_1 + m_2 = 0$ we find $\frac 1 {\sqrt{-m_1^2 m_3^2}}$.  The Hessian above produces the same prefactor since it arises as $\frac 1 {\sqrt{\det H}} \propto \frac 1 {\sqrt{-m_1^2 m_2^2}}$.

For the third case we have
\begin{equation}
  (m_1 - m_2 + m_3)^2 = (m_1 + m_2 - m_3)^2
\end{equation}
which implies either $m_1 = 0$ or $m_2 = m_3$.

If $m_1 = 0$ we have the following critical point solutions
\begin{gather}
  p = (m_2 - m_3, 0), \qquad
  q_1^* = (0, 0), \qquad
  q_2^* = (m_2, 0), \qquad
  q_3^* = (-m_3, 0), \\
  \alpha_1^* = 1, \qquad
  \alpha_2^* = 0, \qquad
  \alpha_3^* = 0, \qquad
  \frac {\alpha_2^*}{\alpha_3^*} = -\frac {m_3}{m_2}.
\end{gather}

% case 3, m_1 = 0
Plugging into the expression for the Hessian matrix we find
\[
  \frac H 2 = \left(
    \begin{array}{cc|cc}
      \eta & 0 & \begin{smallmatrix} -m_3 \\ 0 \end{smallmatrix} & \begin{smallmatrix} -m_3 \\ 0 \end{smallmatrix} \\
      0 & 0 & \begin{smallmatrix} -m_3 \\ 0 \end{smallmatrix} & \begin{smallmatrix} m_2 - m_3 \\ 0 \end{smallmatrix} \\
      \hline
      \begin{smallmatrix} -m_3 & 0\end{smallmatrix} & \begin{smallmatrix} - m_3 & 0 \end{smallmatrix} & 0 & 0 \\
      \begin{smallmatrix} -m_3 & 0\end{smallmatrix} & \begin{smallmatrix} m_2 - m_3 & 0 \end{smallmatrix} & 0 & 0
    \end{array}
  \right)
\]

We have $\det H = 0$.  This eventually produces a shift in the Landau exponent and leads to a singularity behavior $(m_1^2)^{-\frac 1 4}$.

The other solution is $m_2 = m_3$.  At the critical point corresponding to this solution we have
\begin{gather}
  p = (m_1, 0), \qquad
  q_1^* = (m_1, 0), \qquad
  q_2^* = (-m_2, 0), \qquad
  q_3^* = (m_2, 0), \\
  \alpha_1^* = 1, \qquad
  \alpha_2^* = -\frac {m_1}{m_2}, \qquad
  \alpha_3^* = \frac {m_1}{m_2}.
\end{gather}

% case 3, m_2 = m_3
Plugging into the general expression for the Hessian matrix we obtain
\[
  \frac H 2 = \left(
    \begin{array}{cc|cc}
      \bigl(1 + \frac {m_1}{m_2}\bigr) \eta & \frac {m_1}{m_2} \eta & \begin{smallmatrix} m_1 + m_2 \\ \end{smallmatrix} & \begin{smallmatrix} m_2 \\ 0 \end{smallmatrix} \\
      \frac {m_1}{m_2} \eta & 0 & \begin{smallmatrix} m_2 \\ 0 \end{smallmatrix} & \begin{smallmatrix} 0 \\ 0 \end{smallmatrix} \\
      \hline
      \begin{smallmatrix} m_1 + m_2 & 0 \end{smallmatrix} & \begin{smallmatrix} m_2 & 0 \end{smallmatrix} & 0 & 0 \\
      \begin{smallmatrix} m_2 & 0 \end{smallmatrix} & \begin{smallmatrix} 0 & 0 \end{smallmatrix} & 0 & 0
    \end{array}
  \right)
\]

We have $\det \frac H 2 = -m_1^2 m_2^2$ which contributes a factor $\frac 1 {\sqrt{-m_1^2 m_2^2}}$ to the prefactor of the integral.  The computed prefactor $\frac 1 {\sqrt{m_1 m_2 m_3 M_3}}$ yields $\frac 1 {\sqrt{m_1^2 m_2^2}}$ when setting $m_3 = m_2$.

In the fourth case we have
\[
  (m_1 - m_2 - m_3)^2 - (m_1 + m_2 - m_3)^2 = 0.
\]
This implies that either $m_2 = 0$ or $m_1 = m_3$.

In the case $m_2 = 0$ we have the following critical point solutions
\begin{gather}
  p = m_1 - m_3, \qquad
  q_1^* = (m_1, 0), \qquad
  q_2^* = (0, 0), \qquad
  q_3^* = (-m_3, 0), \\
  \alpha_1^* = 0, \qquad
  \alpha_2^* = 1, \qquad
  \alpha_3^* = 0, \qquad
  \frac {\alpha_1^*}{\alpha_3^*} = -\frac {m_3}{m_1}.
\end{gather}

% case 4, m_2 = 0
Plugging into the expression for the Hessian matrix we find
\[
  \frac H 2 = \left(
    \begin{array}{cc|cc}
      0 & 0 & \begin{smallmatrix} m_1 - m_3 \\ 0\end{smallmatrix} & \begin{smallmatrix} -m_3 \\ 0\end{smallmatrix} \\
      0 & \eta & \begin{smallmatrix} -m_3 \\ 0\end{smallmatrix} & \begin{smallmatrix} -m_3 \\ 0\end{smallmatrix} \\
      \hline
      \begin{smallmatrix} m_1 - m_3 & 0\end{smallmatrix} & \begin{smallmatrix} -m_3 & 0\end{smallmatrix} & 0 & 0 \\
      \begin{smallmatrix} -m_3 & 0\end{smallmatrix} & \begin{smallmatrix} -m_3 & 0\end{smallmatrix} & 0 & 0
    \end{array}
  \right)
\]

We have $\det H = 0$ and this ultimately produces a shift of the Landau exponent and a singularity of type $(m_2^2)^{-\frac 1 4}$.

Finally, when $m_1 = m_3$ we have the critical point solutions
\begin{gather}
  p = (m_2, 0), \qquad
  q_1^* = (-m_1, 0), \qquad
  q_2^* = (m_2, 0), \qquad
  q_3^* = (m_1, 0), \\
  \alpha_1^* = -\frac {m_2}{m_1}, \qquad
  \alpha_2^* = 1, \qquad
  \alpha_3^* = \frac {m_2}{m_1}.
\end{gather}

% case 4, m_1 = m_3
Plugging into the general expression for the Hessian matrix yields
\[
  \frac H 2 = \left(
    \begin{array}{cc|cc}
      0 & \frac {m_2}{m_1} \eta & \begin{smallmatrix} 0\\ 0\end{smallmatrix} & \begin{smallmatrix} m_1 \\ 0\end{smallmatrix} \\
      \frac {m_2}{m_1} \eta & \bigl(1 + \frac {m_2}{m_1}\bigr) \eta & \begin{smallmatrix} m_1\\ 0\end{smallmatrix} & \begin{smallmatrix} m_1 + m_2\\ 0\end{smallmatrix} \\
      \hline
      \begin{smallmatrix} 0 & 0\end{smallmatrix} & \begin{smallmatrix} m_1 & 0\end{smallmatrix} & 0 & 0 \\
      \begin{smallmatrix} m_1 & 0\end{smallmatrix} & \begin{smallmatrix} m_1 + m_2 & 0\end{smallmatrix} & 0 & 0
    \end{array}
  \right)
\]

We have $\det \frac H 2 = -m_1^2 m_2^2$.  This agrees with the explicit computation of the prefactor since when $m_3 = m_1$ we have $\frac 1 {\sqrt{m_1 m_2 m_3 M_3}} = \frac 1 {\sqrt{m_1^2 m_2^2}}$.

\section{More Examples}

\subsection{The Massless Box}
\label{sec:massless-box}

Consider a massless box integral with external null momenta $p_i$ and internal momenta $q_i$, $i = 1, \dotsc, 4$ and momentum conservation $p_1 = q_1 - q_4$, $p_2 = q_2 - q_1$, $p_3 = -q_3 - q_2$, $p_4 = q_3 - q_4$.

First, the box integral satisfies the tadpole Landau equations.  These equations are not usually considered but, in the massive case, they control the singularity $m \to 0$.  The tadpole Landau loop equations read $\alpha_i q_i = 0$ (no summation on $i$) and can be solved by $\alpha_i = 1$ and $q_i = 0$.  This happens for all values of the external kinematics and corresponds to a permanent pinch.  In turn, this can translate into IR divergences (see ref.~\cite{10.1063/1.1724268}).

Next, we have bubble Landau singularities, in the $s$ and $t$ channels.  Their corresponding singularities are $s = 0$ and $t = 0$.  In this case the loop momenta at the singularity are proportional to the external momentum but the $\alpha$ (or the momentum fractions) are not uniquely determined.  This is not a simple pinch but a more complicated (and less studied) type of singularity.

Consider the triangle Landau singularity with $q_4$ contracted (the other cases can be derived by symmetry).  The Landau loop equations read $\alpha_1 q_1 + \alpha_2 q_2 - \alpha_3 q_3 = 0$.  Dotting with $q_1$, $q_2$ and $-q_3$ we find
\begin{equation}
  \begin{pmatrix}
    0 & 0 & -\frac t 2 \\
    0 & 0 & 0 \\
    -\frac t 2 & 0 & 0 \\
  \end{pmatrix}
    \begin{pmatrix} \alpha_1 \\ \alpha_2 \\ \alpha_3\end{pmatrix} = 0.
\end{equation}
If we take $t \neq 0$ then we have $\alpha_1 = \alpha_3 = 0$, $\alpha_2 = 1$ and therefore $q_2 = 0$.

Finally, consider the box singularity.  Solving the box Landau equations $\alpha_1 q_1 + \alpha_2 q_2 - \alpha_3 q_3 - \alpha_4 q_4 = 0$ by dotting them with $q_1$, $q_2$, $-q_3$ and $-q_4$ we obtain
\begin{equation}
  \begin{pmatrix}
    0 & 0 & -\frac t 2 & 0 \\
    0 & 0 & 0 & \frac s 2 \\
    -\frac t 2 & 0 & 0 & 0 \\
    0 & \frac s 2 & 0 & 0
  \end{pmatrix}
  \begin{pmatrix} \alpha_1 \\ \alpha_2 \\ \alpha_3 \\ \alpha_4\end{pmatrix} = 0.
\end{equation}
The condition for these equations to have non-trivial solutions in $\alpha$ is the vanishing of the determinant which implies $s t = 0$.  When $s = 0$ we have $\alpha_1, \alpha_3$ undetermined and $\alpha_2 = \alpha_4 = 0$.  When $t = 0$ we have $\alpha_2, \alpha_4$ undetermined and $\alpha_1 = \alpha_3 = 0$.

We emphasize that even though two of the $\alpha$ vanish, this does not mean that we consider the singularity of the contracted graph, with the edges corresponding to the vanishing $\alpha$ being contracted.  In particular, we still impose the on-shell conditions for the edges with vanishing $\alpha$ (see also the discussion ref.~\cite{pham1968singularities}).

Note that the Landau locus of the box integral is not an irreducible variety; it splits into $s = 0$ and $t = 0$.  In particular, it has singularities at the same locations as the $s$- and $t$-channel bubble singularities.  The box singularity appears to be a hierarchically allowed singularity of the $s$- and $t$-channel cut integrals.  This feature could in principle explain naive violations of the Steinmann rule which forbids cuts in overlapping channels; in this case the singularities would arise from a hierarchical singularity, not from an overlapping channel.

Let us now examine the internal kinematics at the location of the box singularity.  The on-shell conditions $q_i^2 = 0$ can be solved (over the complex numbers) in two different ways, since momentum conservation at the three-point vertices can be solved in two ways (either the $\lambda$ or the $\tilde{\lambda}$ are proportional at a vertex).

Picking the vertex at momentum $p_1$ to have $\tilde{\lambda}$ proportional, we have
\begin{gather}
  q_1 = \beta_1 \lambda_2 \tilde{\lambda}_1, \qquad
  q_2 = \beta_2 \lambda_2 \tilde{\lambda}_3, \qquad
  q_3 = \beta_3 \lambda_4 \tilde{\lambda}_3, \qquad
  q_4 = \beta_4 \lambda_4 \tilde{\lambda}_1,
\end{gather}
where
\begin{gather}
  \beta_1 = \frac {\langle 1 4\rangle}{\langle 2 4\rangle} = -\frac {[23]}{[13]}, \qquad
  \beta_2 = \frac {[12]}{[13]} = -\frac {\langle 3 4\rangle}{\langle 2 4\rangle}, \\
  \beta_3 = -\frac {\langle 2 3\rangle}{\langle 2 4\rangle} = \frac {[14]}{[13]}, \qquad
  \beta_4 = -\frac {\langle 1 2\rangle}{\langle 2 4\rangle} = \frac {[34]}{[13]}.
\end{gather}
In complex kinematics, we can satisfy $s \to 0$ either by $\langle 1 2 \rangle = 0$ or $[1 2] = 0$.  The choice $\langle 1 2\rangle = 0$ and the choice of internal kinematics impose $[3 4] = 0$.  Plugging in these values we obtain
\begin{gather}
  \beta_1 = 1, \qquad
  \beta_3 = 1, \qquad
  \beta_4 = 0.
\end{gather}
This means that momentum $q_4$ vanishes.

Similarly, we could take $[12] = 0$ (which is the other way to satisfy the condition $s = 0$) and the corresponding condition $\langle 3 4\rangle = 0$.  In this case, we have $q_2 = 0$ instead.

A similar analysis can be done for the case where the $\lambda$ (instead of the $\tilde{\lambda}$) are proportional at the vertex where $p_1$ is attached and also when analyzing the $t = 0$ singularity.  In all these cases the conclusion is that there is one internal momentum which vanishes at the location of the box Landau singularity.

In conclusion, we have encountered several types of singularities
\begin{enumerate}
\item permanent pinches, which arise for all values of external kinematics.  These are the tadpoles and the triangle singularities.
\item pinches which are not simple (bubble singularities), that is they don't happen at a single point in the integration domain.
\item Landau singularities which yield reducible varieties like $s t = 0$ for the box singularity.
\item coincident singularities: the $q = 0$ singularity can arise for the tadpole, for a special value in the bubble singularity, for the triangle singularity and for the box singularity.
\end{enumerate}

The last case does not appear to have been studied in the mathematical literature, perhaps due to the degree of fine-tuning it requires.  The constructions of refs.~\cite{pham,pham2011singularities} have not been done for this case.  It is tempting to conjecture that in such cases we obtain a behavior of $\log^2 s$ type (see the end of sec.~\ref{sec:rational-monodromy}, remark~\ref{rem:log-squared} for a parallel discussion for iterated integrals).  Nevertheless, integrals with $\log^2$ behavior are very interesting for physics and have been studied extensively (see refs.~\cite{Sudakov:1954sw, Collins:1980ih}).

\subsection{A Triangle Integral}
\label{sec:a-triangle-int}

In this section we will study the triangle integral in fig.~\ref{fig:a-triangle-int} in four dimensions
\begin{equation}
  \label{eq:a-triangle-int}
  I = \int \frac {d^4 q_1}{(q_1^2 - m^2) (q_2^2 - m^2) (q_3^2 - m^2)}.
\end{equation}

\begin{figure}
  \centering
  \includegraphics{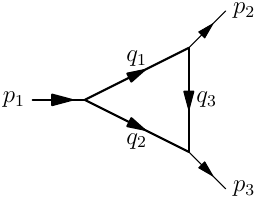}
  \caption{A triangle integral with one external massive leg and internal propagators of equal mass $m$.}
  \label{fig:a-triangle-int}
\end{figure}

This integral is interesting for several reasons.  First, it is a finite integral so it does not require a regularization.  Second, it can be computed explicitly in terms of simple functions so its singularity structure can be verified explicitly.  Third, this integral is an example of integral whose parameters are ``fine-tuned'' so the usual results quoted in the mathematical physics literature about the structure of singularities do not apply in this case.  Finally, thought as a Landau diagram, this appears in the calculation of the double box of ref.~\cite{Caron-Huot:2014lda} (which was also studied in ref.~\cite{Hannesdottir:2024hke} from the point of view of the Landau bootstrap).

Without much loss of generality we pick a frame where\footnote{Since one of the momentum components is $\sqrt{s}$, it will pick up a sign under a monodromy $s = e^{2 \pi i} s$, however this corresponds to a parity symmetry on the momentum components.}
\begin{gather}
  p_1 = (\sqrt{s}, \mathbf{0}), \\
  p_2 = (\frac {\sqrt{s}} 2, \frac {\sqrt{s}} 2, \mathbf{0}), \\
  p_3 = (\frac {\sqrt{s}} 2, -\frac {\sqrt{s}} 2, \mathbf{0}),
\end{gather}
which satisfies $p_1^2 = s$, $p_2^2 = p_3^2 = 0$ and $p_1 = p_2 + p_3$.

We have three tadpole Landau singularities $q_e^2 = m^2$ and $\alpha_e q_e = 0$ which implies $q_e = 0$ and $m = 0$.  Next, we have three bubble Landau singularities.  The simplest one is the $q_1, q_2$ Landau singularity which leads to a Landau singularity $p_1^2 = s = 4 m^2$ in the physical region and $p_1^2 = s = 0$, a pseudothreshold singularity outside of the physical region.  The Landau singularity at $s = 4 m^2$ is a simple pinch where $q_1 = q_2 = \frac 1 2 p_1$ so the loop momenta are uniquely determined.  However, the pseudothreshold singularity is \emph{not} of simple pinch type.  Indeed, solving the Landau equations yields $q_1 + q_2 = 0$, $p_1 = 0$ but the loop momenta $q_1$ and $q_2$ are not uniquely determined.

There are also two bubble Landau singularities in a massless channel.  For definiteness we choose one with momenta $q_1, q_3$.  The Landau equations are $\alpha_1 q_1 = \alpha_3 q_3$.  With $q_1 - q_3 = p_2$ and $p_2^2 = 0$.  Here we have a singularity in the physical region when $m = 0$, but it is not of simple pinch type.  Also, we have a second type singularity at $s = 0$.

For the triangle Landau locus we need to solve the Landau loop equation $\alpha_1 q_1 - \alpha_2 q_2 + \alpha_3 q_3 = 0$ together with the on-shell constraints $q_e^2 = m^2$.  If we dot the Landau loop equation into $q_1$, $-q_2$ and $q_3$ respectively, we obtain the linear system
\begin{equation}
    \begin{pmatrix}
        m^2 & \frac{1}{2} \left(2 m^2-s\right) & m^2 \\
 \frac{1}{2} \left(2 m^2-s\right) & m^2 & m^2 \\
 m^2 & m^2 & m^2
    \end{pmatrix}
    \begin{pmatrix}
        \alpha_1 \\ \alpha_2 \\ \alpha_3
    \end{pmatrix} =
    \begin{pmatrix}
        0 \\ 0 \\ 0
    \end{pmatrix},
\end{equation}
which has a non-trivial solution only when the determinant of the $3 \times 3$ matrix vanishes.  This determinant is $-\frac 1 4 s m^2$.  There are two solutions, which correspond to distinct internal configurations.  The $s = 0$ solution is a second type singularity where $\alpha_1 + \alpha_2 + \alpha_3 = 0$.  The $m = 0$ solution corresponds to $\alpha_1 = \alpha_2 = 0$ and $\alpha_3 = 1$.

The Landau singularity in the physical region at $m = 0$ occurs at $q_1 = p_2$, $q_2 = p_3$ and $q_3 = 0$.  This has a unique solution, but which is at the boundary of the set of solutions for the bubble cut in the massless channels.  In favorable cases the singularities can be separated, but that is not the case here.

\section{Monodromy}

\subsection{Monodromy for iterated integrals of rational differential forms}
\label{sec:rational-monodromy}

The following result is very important when computing the monodromy of an iterated integral around a branch point (see ref.~\cite[prop.~1.4]{goncharov2001multiplepolylogarithmsmixedtate}).

\begin{theorem}
\label{thm:iterated_residue}
Consider an iterated integral with forms $\omega_1, \dotsc, \omega_l$, such that the form $\omega_p$ has a pole along a codimension one variety $S$ and no other forms have a singularity there. Next, consider two paths $\gamma_\pm$ with the same end points and such that they go around $S$ in opposite ways such that $\gamma_+ \gamma_-^{-1}$ goes around $S$ in the counter-clockwise orientation.  Then, we have
\begin{multline}
  \label{eq:iterated_residue}
    \int_{\gamma_+} \omega_1 \circ \cdots \circ \omega_l - \int_{\gamma_-} \omega_1 \circ \cdots \circ \omega_l =\\ 2 \pi i\operatorname{res} \omega_p \int_{\gamma'} \omega_1 \circ \cdots \circ \omega_{p-1} \int_{\gamma''} \omega_{p + 1} \circ \cdots \circ \omega_l,
\end{multline}
where $\gamma'$ is the initial section of the path until $S$ and $\gamma''$ is the final section of the path $\gamma$ starting at $S$ and ending at the end-point of $\gamma$.
\end{theorem}

\begin{proof}
We can decompose $\gamma_+ = \gamma' C_+ \gamma''$ and $\gamma_- = \gamma' C_- \gamma''$, where $C_\pm$ is are half-circles of small radius.  We will send this radius to zero at the end.

We will use the following two properties
\begin{gather}
    \label{eq:reversal}
    \int_{\gamma} \omega_1 \circ \cdots \circ \omega_m = (-1)^m \int_{\gamma^{-1}} \omega_m \circ \cdots \circ \omega_1, \\
    \label{eq:concatenation}
    \int_{\gamma_1 \gamma_2} \omega_1 \circ \cdots \circ \omega_m = \sum_{i = 0}^m \int_{\gamma_1} \omega_1 \circ \cdots \circ \omega_i \int_{\gamma_2} \omega_{i + 1} \circ \cdots \circ \omega_m
\end{gather}
When using the second property to expand the iterated integrals over $\gamma_\pm = \gamma' C_\pm \gamma''$ the terms where there are multiple differential forms on the $C_\pm$ section of the path will vanish in the limit of zero radius.  Indeed, if there are $q$ one-forms along $C_\pm$ the integration is over a $q$-simplex whose volume is proportional to the $q$-th power of the radius.  The only way for the integral not to vanish in this limit is to keep only the one-form which has a pole there.

Finally, using the fact that $\gamma_+ \gamma_-^{-1} = C_+ C_-^{-1}$ is oriented counter-clockwise, and the two properties above, we have
\begin{equation}
    \int_{C_+} \omega_p - \int_{C_-} \omega_p =
    \int_{C_+ C_-^{-1}} \omega_p =
    2 \pi i \operatorname{res} \omega_p.
\end{equation}

The expected result follows immediately.
\end{proof}

\begin{remark}
\label{rmk:independence}
The path $\gamma'$ ends on the polar variety $S$ while $\gamma''$ starts there.  By integrability we can move along $S$ the point around which we take the circle $C_+ C_-^{-1}$.  Said differently, the right hand side expression in Theorem.~\ref{thm:iterated_residue} does not depend on the point on $S$ where the residue is evaluated and where the path $\gamma'$ ends and $\gamma''$ begins.
\end{remark}
\begin{proof}
To illustrate the independence on the location on the polar locus, consider a length two iterated integral $\int_\gamma \omega_1 \circ \omega_2$ and take the difference
\begin{equation}
    \int_{\gamma_+} \omega_1 \circ \omega_2 - 
    \int_{\gamma_-} \omega_1 \circ \omega_2 =
    2 \pi i \operatorname{res} \omega_2 \int_{\gamma'} \omega_1.
\end{equation}
We can take a differential of the right hand side with respect to the end point of the path $\gamma'$.  This point belongs to the polar locus $S$.  When taking the differential we obtain two terms
\begin{equation}
    (\operatorname{res} \omega_2) \omega_1 + (d \operatorname{res} \omega_2) \int_{\gamma'} \omega_1.
\end{equation}
The first term vanishes since it can be obtained by taking the residue of the integrability condition $\omega_1 \wedge \omega_2 = 0$.  The second term vanishes since the residue of a closed form is a closed form.  Indeed, if $\omega$ has a pole of order one along $S$ and $S$ has equation $s = 0$, then we can write
\begin{equation}
    \omega = \rho \wedge \frac{d s}{s} + \sigma,
\end{equation}
where $\rho$, $\sigma$ are regular forms (they have no singularities on $S$) and $\rho$ is closed.  Indeed, if $d \rho \neq 0$ when restricted to $S$, would mean that $d \omega$ has a pole along $S$.  This is absurd since $d \omega = 0$.
\end{proof}

\begin{remark}\label{rem:log-squared}
This theorem can be generalized to the case when several of the forms $\omega$ have a pole along the same variety $S$.  The basic idea is that if a number of one-forms have a pole at the same location, then their iterated integral over a very small circle yields a non-vanishing contribution in the limit where the radius of the circle is sent to zero.  For example
\[
    \lim_{\epsilon \to 0} \int_{C_\epsilon} \omega_1 \circ \omega_2 = \frac 1 2 (2 \pi i)^2 \operatorname{res} \omega_1 \operatorname{res} \omega_2.
\]
To show this, it is enough to take $\omega_1 = \omega_2 = \frac {d z} z$.  After a change of coordinates $z = e^{i \theta}$ we have $\omega_1 = \omega_2 = i d \theta$ and the iterated integral can be computed easily
\[
    \int_{0}^{2 \pi} d \theta \circ d \theta = \int_0^{2 \pi} d \theta_1 \int_{\theta_1}^{2 \pi} d \theta_2 = \frac 1 2 (2 \pi)^2.
\]
\end{remark}

We want to emphasize that the result of theorem~\ref{thm:iterated_residue} implies that an iterated integral has a monodromy (and therefore a singularity) not only in the first entry but in all entries.  Imposing the conditions on which of these singularities are physical (or visible) and which ones are not constrains the terms proportional to powers of $\pi$ (which have been called ``beyond the symbol'' terms).

\subsection{Monodromy with square roots}
\label{sec:monodromy-sqrt}

\begin{figure}
  \centering
  \begin{subfigure}[b]{0.45\textwidth}
    \centering
    \includegraphics[width=\textwidth]{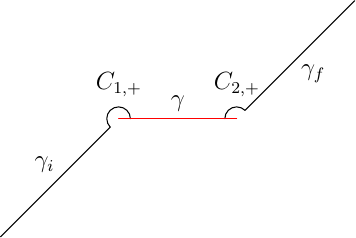}
    % \caption{}
    % \label{}
  \end{subfigure}
  \hfill
  \begin{subfigure}[b]{0.45\textwidth}
    \centering
    \includegraphics[width=\textwidth]{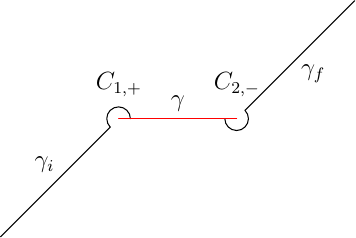}
    % \caption{}
    % \label{}
  \end{subfigure}
  \vspace{\baselineskip}
    \begin{subfigure}[b]{0.45\textwidth}
    \centering
    \includegraphics[width=\textwidth]{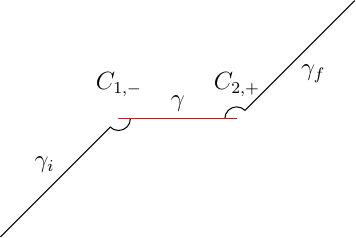}
    % \caption{}
    % \label{}
  \end{subfigure}
  \hfill
  \begin{subfigure}[b]{0.45\textwidth}
    \centering
    \includegraphics[width=\textwidth]{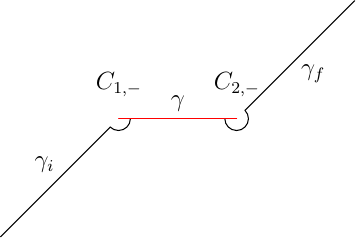}
    % \caption{}
    % \label{}
  \end{subfigure}
  \caption{The possible ways of avoiding two square root branch points.  The branch cut joining the two branch points is depicted in red.}
  \label{fig:sqrt-contours}
\end{figure}

Consider now a different situation than in sec.~\ref{sec:rational-monodromy} where we have two square root branch points (see fig.~\ref{fig:sqrt-contours}).  As can be seen in fig.~\ref{fig:sqrt-contours}, there are four different contours depending on how the branch points are avoided.  We denote them as follows
\begin{gather}
  \gamma_{++} = \gamma_i C_{1,+} \gamma C_{2,+} \gamma_f, \\
  \gamma_{+-} = \gamma_i C_{1,+} \gamma C_{2,-} \gamma_f, \\
  \gamma_{-+} = \gamma_i C_{1,-} \gamma C_{2,+} \gamma_f, \\
  \gamma_{++} = \gamma_i C_{1,-} \gamma C_{2,-} \gamma_f,
\end{gather}
where $\gamma_i$ is the initial part of the contour from the initial point to the first branch point, $C_{1,\pm}$ are the arcs around the first branch point, going above or below, respectively, and similarly for $C_{2,\pm}$.  Finally, $\gamma$ is the contour between the arc around the first branch point and the arc around the second branch point while $\gamma_f$ is the final part of the contour ending at the endpoint of the path.  This contour, together with differential one-forms $\omega_1, \dotsc, \omega_l$, defines a function $F(x)$, where $x$ is the coordinate of the endpoint of the path via the following iterated integral
\begin{equation}
  F(x) = \int_{\gamma_{++}} \omega_1 \circ \cdots \circ \omega_l.
\end{equation}
In practice it may happen that $\gamma_i$ is missing and the contour starts at one of the square root branch points.

The monodromy around the first branch point is given by
\begin{equation}
  (1 - M_1) F(x) = \int_{\gamma_{++}} \omega_1 \circ \cdots \circ \omega_l -
  \int_{\gamma_{-+}} \omega_1 \circ \cdots \circ \omega_l,
\end{equation}
while the monodromy around the second branch point is
\begin{equation}
  (1 - M_2) F(x) = \int_{\gamma_{++}} \omega_1 \circ \cdots \circ \omega_l -
  \int_{\gamma_{+-}} \omega_1 \circ \cdots \circ \omega_l.
\end{equation}

The general form of the differential forms we will study is
\begin{equation}
  \label{eq:diff-form-2sqrt}
  \omega_i = d \log \frac {\sqrt{x - a_i} + \sqrt{x - b_i}}{\sqrt{x - a_i} - \sqrt{x - b_i}} = \frac {d x}{\sqrt{x - a_i} \sqrt{x - b_i}}.
\end{equation}

The integrals around the arcs of circles vanish in the limit of zero radius; the integrand behaves like the inverse of square root, which is an integrable singularity.  For this reason we will drop the arcs of circle in the following.

However, going around one of these small circles from below, has the effect of changing the sign of a differential form which has a square root branch point there.  As can be found from Eq.~\eqref{eq:diff-form-2sqrt},
\begin{equation}
  M_w \omega_i = \begin{cases}
    -\omega_i, &\qquad \text{if $w = a_i$ or $w = b_i$}, \\
    \omega_i, &\qquad \text{otherwise}.
  \end{cases}
\end{equation}

Using the concatenation property (Eq.~\eqref{eq:concatenation}) we have
\begin{multline}
  \label{eq:sqrt-monodromy}
  (1 - M_1) F(x) =
  \sum_{k = 0}^l \int_{\gamma_i} \omega_1 \circ \cdots \circ \omega_{k} \times \\\Bigl(
  \int_{\gamma \gamma_f} \omega_{k + 1} \circ \cdots \circ \omega_l -
  \int_{\gamma \gamma_f} M_1(\omega_{k + 1}) \circ \cdots \circ M_1(\omega_l)\Bigr).
\end{multline}
Obviously, only the terms for which an odd number of signs appear from $M_i(\omega_j)$ contribute to the answer.  A similar formula holds for $M_2$.  This is the square root analog of Eq.~\eqref{eq:iterated_residue}.

\subsection{An example}
\label{sec:example}

Let us now apply the methods of sec.~\ref{sec:monodromy-sqrt} to an example.  Consider the function
\begin{equation}
  F(x) = \frac 1 2 \log^2 \frac {\sqrt{x - a} + \sqrt{x - b}}{-\sqrt{x - a} + \sqrt{x - b}},
\end{equation}
which is really the integral of sec.~\ref{sec:a-triangle-int}.  We have $F(a) = 0$ which means we can think of $F(x)$ as an iterated integral along a path which starts at $a$ and ends at $x$.

We have
\begin{equation}
  F(x) = \int_{\gamma_{++}} \omega \circ \omega,
\end{equation}
where $\gamma_{++}$ is the piecewise linear path with arcs of circle around the square root branch points $x = a$ and $x = b$ (as in fig.~\ref{fig:sqrt-contours}) and such that $\gamma_i$ is empty.  We also have\footnote{Here we treat $a_i$ and $b_i$ as constants.  In practice, we are dealing with multi-variable functions and $a_i$, $b_i$ depend on other variables which produce extra contributions (such as poles) when taking the differential.  We will ignore these complications for now, to keep the presentation simple.}
\begin{equation}
  \omega =
  d \log \frac {\sqrt{x - a} + \sqrt{x - b}}{\sqrt{x - a} - \sqrt{x - b}} =
  d \log \frac {\sqrt{x - a} + \sqrt{x - b}}{-\sqrt{x - a} + \sqrt{x - b}} =
  \frac {d x}{\sqrt{x - a} \sqrt{x - b}}.
\end{equation}
Then, using the result in Eq.~\eqref{eq:sqrt-monodromy}, we find that $(1 - M_a) F(x) = 0$, that is the singularity at $x = a$ is not visible.

To compute the monodromy around $x = a$ we have $\gamma_{++} = C_{1, +} \gamma C_{2, +} \gamma_f$ and $\gamma_{-+} = C_{1, -} \gamma C_{2, +} \gamma_f$.  The $C_{1, \pm}$ sectors do not contribute to $\epsilon \to 0$.  Neglecting terms which vanish in the $\epsilon \to 0$ limit we obtain
\begin{equation}
  (1 - M_a) F(x) = \int_{\gamma C_{2,+} \gamma_f} (\omega \circ \omega - M_a(\omega) \circ M_a(\omega)) = 0,
\end{equation}
since $M_a(\omega) = -\omega$.  This is the mechanism by which the pseudo-threshold singularity is invisible on the main sheet.

We also have
\begin{equation}
  (1 - M_b) F(x) = \Bigl(\int_\gamma \omega\Bigr) \Bigl(\int_{\gamma_f} (\omega - M_b(\omega))\Bigr).
\end{equation}
We have $M_b(\omega) = -\omega$.  Also, we can write
\begin{equation}
  \int_\gamma \omega = \frac 1 2 \int_\Gamma \omega = -\pi i,
\end{equation}
where $\Gamma$ is the contour going clockwise around the cut.  This integral can be computed by deforming the contour to pick up the pole at infinity.

We conclude that
\begin{equation}
  (1 - M_b) F(x) = -2 \pi i \log \frac {\sqrt{x - a} + \sqrt{x - b}}{\sqrt{x - a} - \sqrt{x - b}},
\end{equation}
where the factor of $2$ comes from the doubling of the integral along the contour $\gamma_f$.  This new quantity has a singularity at $x = a$, in other words we have
\begin{equation}
  (1 - M_a) (1 - M_b) F(x) = 2 (1 - M_b) F(x) \neq 0.
\end{equation}

In contrast, (for $b > a$) the ``unitarity cut'' is given by
\begin{equation}
  \lim_{\epsilon \to 0^+} (F(x + i \epsilon) - F(x - i \epsilon)) = \theta(x - b) (-2 \pi i) \log \frac {\sqrt{x - a} + \sqrt{x - b}}{\sqrt{x - a} - \sqrt{x - b}}.
\end{equation}
Notice that the signs in the denominator of the argument of the logarithm have changed!  This quantity corresponds to \emph{another} difference
\begin{multline}
  \lim_{\epsilon \to 0^+} (F(x + i \epsilon) - F(x - i \epsilon)) =
  \int_{\gamma_{++}} \omega \circ \omega -
  \int_{\gamma_{--}} \omega \circ \omega = \\
  \int_\gamma (\omega \circ \omega - M_a(\omega) \circ M_a(\omega)) + \\
  \int_\gamma \omega \int_{\gamma_f} \omega - \int_\gamma M_a(\omega) \int_{\gamma_f} M_a M_b(\omega) + \\
  \int_{\gamma_f} (\omega \circ \omega - M_a M_b(\omega) \circ M_a M_b(\omega)) =
  2 \int_\gamma \omega \int_{\gamma_f} \omega,
\end{multline}
where we have used $M_a M_b(\omega) = \omega$.  The integral $\int_\gamma \omega = -\pi i$ can be computed in the same way as before while
\begin{equation}
  \int_{\gamma_f} \omega =
  \left. \log \frac {\sqrt{y - a} + \sqrt{y - b}}{\sqrt{y - a} - \sqrt{y - b}}\right\rvert_{y = b}^{y = x} =
    \log \frac {\sqrt{x - a} + \sqrt{x - b}}{\sqrt{x - a} - \sqrt{x - b}}.
\end{equation}

We end this section with some comments on contours.  In the analysis of sec.~\ref{sec:example} we assumed that the path from $a$ to $x$ encounters $b$ (above the threshold).  This assumes that $b > a$.  In a physics example we have $a = 0$ while $b = m^2$ and $b > a$ as long as $m \in \mathbb{R}$.  But if $m^2 < 0$, then we never encounter a threshold singularity for $x > 0$, if the contour is defined as a piecewise straight line.  This is the mechanism by which singularities get submerged to other sheets when they become non-physical, for example when they correspond to $\alpha < 0$.

\section{Second type Singularities via Inversion}
\label{sec:second_type_inversion}

In this section we describe the second type singularities in a way which is simpler and should allow various (type I) results to apply in this new case.

The idea is to invert the momenta and make a change of variables so that, roughly, $q_e \to \frac{q_e}{q_e^2}$ for each internal momentum in a loop. This way, we expose singularities that occur as pinches of the loop momenta at infinity. It turns out that doing this inversion in \emph{dual coordinates}~\cite{Okun:1960cls, Broadhurst:1993ib, Drummond:2007aua} makes the computations of second-type singularities particularly straightforward as we now show.

\paragraph{Second type singularities at one loop}
Assume that we have an $n$-point one-loop scalar integral
\begin{equation}
    \I_{\text{n-gon}} = \int_h \frac{d^D q_1}{\prod_{e=1}^{n} (q_e^2-m_e^2)} \,,
    \label{eq:I_n-gon}
\end{equation}
where we have allowed for any integration contour $h$ that avoids the singularities at $q_e^2=m_e^2$ for generic values of the external momenta, and thus dropped the explicit $i \varepsilon$'s from the denominators. We have set the numerator $N$ to $1$ for definiteness, but this can easily be relaxed.  

In dual coordinates, we label the loop with $x_0$, and put $q_e = x_e - x_0$ for each internal momentum.  The external masses are then given by $M_e^2=(x_e -x_{e+1})^2$, with the cyclic identification $x_{n+1} \equiv x_{1}$ understood. In these coordinates, the integral becomes
\begin{equation}
\I_{\text{n-gon}} = 
\int_{h'} \frac {d^D x_0}{\prod_{e = 1}^n [(x_e-x_{0})^2 - m_e^2]}.
\end{equation}
We want to study the singularities at infinity, so we do a change of variable to bring the infinity to the origin:
\begin{equation}
x_0 \to \frac {x_0}{x_0^2}.
\end{equation}
\begin{equation}
\I_{\text{n-gon}} = 
\int_{h'}
\frac {d^D x_0}{(x_0^2)^{D - n}} \frac 1 {\prod_{e = 1}^n  [1 - 2 x_0 \cdot x_e + x_0^2 (x_e^2 - m_e^2)]}.
\label{eq:n-gon_secondtype}
\end{equation}
We will assume that $D > n$.  Note that if $D = n$ then there is no $x_0^2$ in the denominator, which is consistent with the absence of second-type singularities for $n$-gons in $n$ dimensions. Moreover, if $D < n$ the $x_0^2$ factor appears in the numerator instead and can not contribute to a pinch.  If $D > n$ we also need to do a blow-up to resolve the singularity at $x_0 = 0$.

Next, we can look at the singularities of the differential form appearing in~\eqref{eq:n-gon_secondtype}, just like we do for usual non-inverted momentum-space integrands. Compared to first-type singularities, we now have an extra denominator factor of $x_0^2$, which acts like a new propagator that can take part in the pinch. The condition for the pinch becomes the existence of $\alpha_0, \alpha_1, \dotsc, \alpha_n$ not all vanishing such that
\begin{equation}
\alpha_0 d x_0^2 + \sum_{e = 1}^n \alpha_e d \left[1 - 2 x_0 \cdot x_e + x_0^2 (x_e^2 - m_e^2) \right] = 0,
\end{equation}
where we set $d x_e = 0$ and $d m_e = 0$.  This means that
\begin{equation}
\sum_{e = 1}^n \alpha_e x_e = x_0 \Bigl[\alpha_0 + \sum_{e = 1}^n \alpha_e (x_e^2 - m_e^2)\Bigr].
\end{equation}
Since we are looking for singularities of second type, we take $x_0^2 = 0$ to be one of the
``on-shell'' conditions. Then, this equation implies that we can form a linear combination of dual coordinates $x_i$ which is equal to a null vector.  In turn, this means that the Gram matrix of these vectors vanishes.

These kinds of conditions occur in ref.~\cite{doi:10.1063/1.1724262} by Fairlie, Landshoff, Nuttall and Polkinghorne (see also the discussion in ref.~\cite[sec.~2.10]{ELOP}).  Their formulations are not straightforward, and involve notions like ``super-pinch''.  The other advantage of our method becomes obvious if we include numerators.  Then the inversion will generate a denominator which depends on the numerator we have and this brings about its own shift in the Landau exponent via Landau--Leray--Pham's asymptotic formulas.

Another advantage is that once we do an inversion we can in principle compute a cut \`a la Cutkosky, by simply replacing the propagators which participate in the pinch by delta functions.

\paragraph{Second type singularities of higher-loop integrals}

Higher loops can be analyzed in a similar way but there are several patches we need to consider (which correspond to mixed first type and second type singularities).  In that case we need to study the possibilities where we invert in all possible or only a subset of the loops.  The fact that there can be multiple singularities at infinity was also described in the older literature (see ref.~\cite{Greenman:1969cj}).

As an example, we take the scalar ice-cream cone integral,
\begin{figure}
  \centering
  \includegraphics{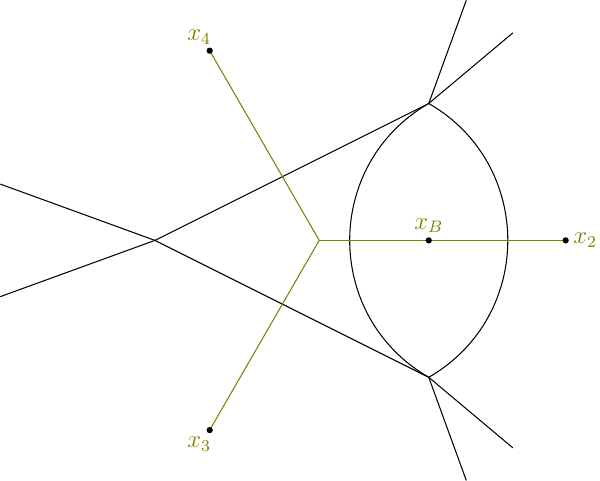}
  \caption{The ice-cream cone integral and its dual space.}
  \label{fig:ice-cream}
\end{figure}
which we write in dual coordinates as
\begin{equation}
    \I_{\text{ice}} = \int_h \frac{d^D x_A d^D x_B}{[(x_A{-}x_B)^2{-}m_1^2][(x_2{-}x_B)^2{-}m_2^2][(x_3{-}x_A)^2{-}m_3^2][(x_4{-}x_A)^2{-}m_4^2]} \,,
\end{equation}
where $h$ is an integration contour that avoids the singularities of the denominators for generic values of $x_2$, $x_3$ and $x_4$.
To capture the first-type singularities of this integral, we solve the Landau equations as usual, and we will not go through that here.

For second-type singularities, we invert the dual-coordinate loop momenta $x_A$ and $x_B$, both separately and at the same time. The first case we look at is when $x_B$ is inverted, but $x_A$ is not. That is, we set
\begin{equation}
    y_B = \frac{x_B}{x_B^2} \,.
\end{equation}
Then,
\begin{align}
    \I_{\text{ice}} & = \int_{h'} \frac{d^D x_A d^D y_B}{(y_B^2)^{D-2}[1{-}2 x_A {\cdot} y_B {+}(x_A^2{-}m_1^2)y_B^2] [1{-}2 x_2 {\cdot} y_B {+}(x_2^2{-}m_1^2)y_B^2]}
    \\ & \hspace{6cm} \times \frac{1}{[(x_3{-}x_A)^2{-}m_3^2][(x_4{-}x_A)^2{-}m_4^2]} \,,
\end{align}
where $h'$ is the contour corresponding to $h$ in the new variables.
We now solve the leading Landau equations of this expression, that is, assuming that all denominator factors are ``on shell'', including the condition $y_B^2=0$. The Landau equations then consist of the on-shell conditions
\begin{align}
    \label{eq:onshell}
    y_B^2=0 \,, \qquad 1{-}2 x_A {\cdot} y_B = 0\,, \qquad 1{-}2 x_2 {\cdot} y_B = 0\,, \\
    (x_3-x_A)^2 = m_3^2 \,, \qquad (x_4-x_A)^2 - m_4^2 \,,
    \nonumber
\end{align}
as well as the loop equations in $x_A$ and $y_B$, which say that we can find $\alpha_e$ not all equal to zero such that
\begin{subequations}
\begin{align}
    \alpha_0 \, y_B -\alpha_1 \left[-x_A + (x_A^2-m_1^2) y_B \right] + \alpha_2 \left[-x_2 + (x_2^2-m_1^2) y_B\right] & = 0
    \label{eq:loop_icecream_1}
    \\
    - \alpha_1 \, y_B + \alpha_3 (x_A-x_3) + \alpha_4 (x_A-x_4) & = 0
    \label{eq:loop_icecream_2}
\end{align}
\end{subequations}
Dotting the first equation with $y_B$ and using the on-shell conditions from~\eqref{eq:onshell} gives $\alpha_1=-\alpha_2$. Plugging it back into~\eqref{eq:loop_icecream_1} gives $\alpha_2 (x_A-x_2)^2 \propto y_B$, so
\begin{equation}
    \alpha_2^2 (x_A-x_2)^2 = 0 \,.
    \label{eq:extra_zero_cond}
\end{equation}
Next, dotting~\eqref{eq:loop_icecream_1} with $\langle x_2 x_3 x_4 \cdot $ (in the notation of Appendix~\ref{sec:landau-geometric}) gives
\begin{equation}
    \langle x_2 x_3 x_4 y_B \rangle = 0 \,.
    \label{eq:yBzero}
\end{equation}
Dotting $\langle x_2 x_3 x_4 \cdot $ into~\eqref{eq:loop_icecream_2} then gives
\begin{equation}
    (\alpha_3+\alpha_4) \langle x_2 x_3 x_4 x_A \rangle = 0 \,,
\end{equation}
where we have used~\eqref{eq:yBzero}. We thus find a solution to the Landau equations for which
\begin{equation}
    \langle x_2 x_3 x_4 x_A \rangle = 0\,.
\end{equation}
We can use the on-shell conditions $(x_3-x_A)^2 = m_3^2$, $(x_4-x_A)^2 = m_4^2$, as well as $ (x_2-x_A)^2 =0$ from~\eqref{eq:extra_zero_cond} to compute
\begin{equation}
    \langle x_2 x_3 x_4 x_A \rangle =
    \frac{1}{4} \left[-m_3^2 \left(m_4^2-p_2^2\right){}^2+\left(m_3^2+m_4^2-p_1^2\right) \left(m_3^2-p_3^2\right) \left(m_4^2-p_2^2\right)-m_4^2 \left(m_3^2-p_3^2\right)^2\right] \,.
\end{equation}
The last thing to check is that this solution actually leads to a sensible solution for the remaining $\alpha_0, \alpha_1, \ldots \alpha_4$, as well as the loop momenta $x_A$ and $y_B$, with no further restrictions on the kinematics.

One can study non-planar integrals by following the recipe sketched in Appendix~\ref{sec:landau-geometric} (see fig.~\ref{fig:nonplanar}).

\section{Acknowledgments}

It is a pleasure to thank H.~Hannesd\'ottir, A.~McLeod and M.~D.~Schwartz for initial collaboration on this paper and for many discussions about Landau singularities.  I am also grateful to J.~Bourjaily and J.~Collins for discussions and advice.

\appendix

\section{A geometric approach to solving Landau equations}
\label{sec:landau-geometric}

In this section we outline a geometric approach to solving the Landau equations.  In general cases, some complications may arise which, for lack of space, we can not describe here in detail.  We only describe here the bare minimum required to describe the solution in sec.~\ref{sec:second_type_inversion}.

For a triangle integral with external momenta $p_1, p_2, p_3$ and internal momenta $q_1, q_2, q_3$ we have, after going to dual space,
\begin{gather}
    p_1 = x_2 - x_3, \qquad
    p_2 = x_3 - x_1, \qquad
    p_3 = x_1 - x_2,
\end{gather}
\begin{gather}
    q_1 = x_A - x_1, \qquad
    q_2 = x_A - x_2, \qquad
    q_3 = x_A - x_3,
\end{gather}
with $p_i^2 = M_i^2$ and $q_i^2 = m_i^2$.

The Landau loop equation is
\begin{equation}
    \alpha_1 q_1 + \alpha_2 q_2 + \alpha_3 q_3 = 0,
\end{equation}
or in dual language
\begin{equation}
    \alpha_1 x_1 + \alpha_2 x_2 + \alpha_3 x_3 = (\alpha_1 + \alpha_2 + \alpha_3) x_A.
\end{equation}
Then we have $\langle 1 2 3 A\rangle = 0$, where the bracket means the signed (or oriented) volume of the simplex with vertices $1, 2, 3, A$, defined as follows.  For an $n$-simplex with vertices $v_1, \dotsc, v_{n + 1}$ and components $v_i = (v_i^0, \dotsc, v_i^{n - 1})$, we define\footnote{We can use inversion for one of the vectors $v_j$ which sends it to $i(v_j) = \frac {v_j}{v_j^2}$.  Then
\[
  \langle v_1, \dotsc, i(v_j), \dotsc, v_n\rangle =
  \frac 1 {(v_j)^2}
  \begin{pmatrix}
    1      & \hdots & (v_j)^2& \hdots & 1 \\
    v_1^0  & \hdots & v_j^0  & \hdots & v_n^0 \\
    \vdots & \vdots & \vdots & \vdots & \vdots \\
    v_1^{n - 1}& \hdots & v_j^{n - 1} & \hdots & v_n^{n - 1}
  \end{pmatrix}.
\]}
\begin{equation}
  \langle v_1, \dotsc, v_{n + 1}\rangle = \det
  \begin{pmatrix}
    1 & 1 & \hdots & 1 \\
    v_1^0 & v_2^0 & \hdots & v_n^0 \\
    \vdots & \vdots & \hdots & \vdots \\
    v_1^{n - 1} & v_2^{n - 1} & \hdots & v_n^{n - 1}
  \end{pmatrix}.
\end{equation}

Alternatively, we can think about lifting these vector identities to the Clifford algebra.  Then instead of brackets we have products in the Clifford algebra followed by picking the degree three parts in the Clifford algebra grading.  If the solutions of the Landau equations can be written using spinors, that would explain the usefulness of momentum twistors and would point the way to their generalization.

Squaring, we find\footnote{This holds for Euclidean signature.  In signature $(p, q)$ with $p + q = n$ and $p$ negative eigenvalues and $q$ positive eigenvalues we have instead
\[
\epsilon(v_1, \dotsc, v_n) \epsilon(w_1, \dotsc, w_n) = (-1)^p \det (v_i \cdot w_j)_{i, j = 1, \dotsc, n}.
\]
Indeed, this identity can be proved by showing that the left and right hand sides have a correct multi-linearity property in the vectors $v$ and $w$, the correct full antisymmetry under permutations of $v$ and $w$ (which amount to permutations of rows and columns in the determinant).  Finally, the numeric coefficient is determined by evaluating both the left and right-hand sides on orthonormal bases.  This is how the extra $(-1)^p$ arises.}
\begin{multline}
    \langle 1 2 3 A\rangle^2 =
    \epsilon(x_{2 1}, x_{3 1}, x_{A 1})^2 = \\
    \det \begin{pmatrix}
        M_3^2 & -\frac{1}{2} (M_1^2 - M_2^2 - M_3^2) & -\frac{1}{2} (m_2^2 - m_1^2 - M_3^2) \\
        & M_2^2 & -\frac{1}{2} (m_3^2 - M_2^2 - m_1^2) \\
        & & m_1^2
    \end{pmatrix} = \\
    -m_1^2 m_2^2 m_3^2 (2 y_{12} y_{13} y_{23} + y_{12}^2 + y_{13}^2 + y_{23}^2 - 1),
\end{multline}
where $y_{12} = \frac{M_3^2 - m_1^2 - m_2^2}{2 m_1 m_2}$ and cyclic permutations.

It can be checked explicitly by area decomposition (see sec.~\ref{sec:volume_decomposition}) that $\langle 1 2 3\rangle = \langle A 2 3\rangle + \langle 1 A 3\rangle + \langle 1 2 A\rangle$.  If we plug the Landau loop equation in dual language in $\langle \cdot 2 3\rangle$ we find $\alpha_1 \langle 1 2 3\rangle = \langle A 2 3\rangle$ and therefore
\begin{gather}
    \alpha_1 = \frac{\langle A 2 3\rangle}{\langle 1 2 3\rangle}, \qquad
    \alpha_2 = \frac{\langle 1 A 3\rangle}{\langle 1 2 3\rangle}, \qquad
    \alpha_3 = \frac{\langle 1 2 A\rangle}{\langle 1 2 3\rangle}.
\end{gather}
Indeed then we have $\alpha_1 + \alpha_2 + \alpha_3 = 1$.

The case $\alpha_1 + \alpha_2 + \alpha_3 = 0$ is special; indeed, contracting the Landau loop equation with $\langle \cdot 2 3\rangle$ we obtain $\alpha_1 \langle 1 2 3\rangle = (\alpha_1 + \alpha_2 + \alpha_3) \langle A 2 3\rangle$ and $\alpha_1 + \alpha_2 + \alpha_3 = 0$ implies that $\langle 1 2 3\rangle = 0$.  Geometrically, this means that the triangle formed of external momenta is degenerate.

We have
\begin{gather}
    \langle 1 2 3\rangle^2 = -\frac{1}{4} (M_1^4 + M_2^4 + M_3^4 - 2 M_1^2 M_2^2 - 2 M_1^2 M_3^2 - 2 M_2^2 M_3^2), \\
    \langle A 2 3\rangle^2 = -m_2^2 m_3^2 (y_{23}^2 - 1), \\
    \langle 1 A 3\rangle^2 = -m_1^2 m_3^2 (y_{13}^2 - 1), \\
    \langle 1 2 A\rangle^2 = -m_1^2 m_2^2 (y_{12}^2 - 1).
\end{gather}

We now discuss briefly, on an example, some special features that arise in non-planar cases.  Consider a two-loop non-planar integral as in fig.~\ref{fig:nonplanar}.

\begin{figure}
    \centering
    \includegraphics{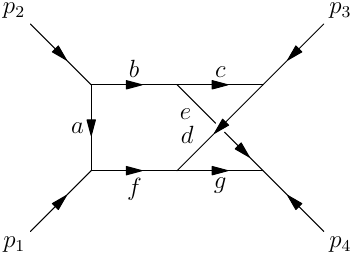}
    \caption{A two-loop non-planar integral.}
    \label{fig:nonplanar}
\end{figure}

Momentum conservation reads
\begin{gather}
    p_2 - a - b = 0, \qquad
    b - c - e = 0, \\
    p_1 + a - f = 0, \qquad
    d + f - g = 0, \\
    p_3 + c - d = 0, \qquad
    p_4 + e + g = 0.
\end{gather}

We introduce a dual space such that
\begin{gather}
    p_1 = x_2 - x_1, \qquad
    p_2 = x_3 - x_2, \\
    p_3 = x_4 - x_3, \qquad
    p_4 = x_1 - x_4, \\
    f = x_A - x_1, \qquad
    a = x_A - x_2, \\
    b = x_3 - x_A, \qquad
    c = x_3 - x_B, \\
    d = x_4 - x_B, \qquad
    e = x_B - x_A, \\
    g = x_B' - x_1
\end{gather}
with the constraint that $x_B' - x_A = x_4 - x_B$.  This is a difference with respect to the planar case; we now have two dual points $x_B$ and $x_B'$ corresponding to one loop.

\begin{figure}
    \centering
    \includegraphics{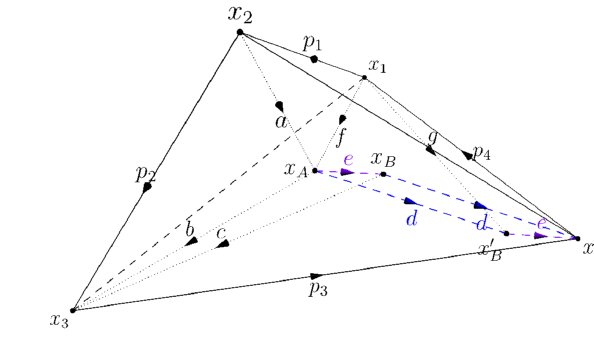}
    \caption{The non-planar integral in fig.~\ref{fig:nonplanar} in dual space.  Since the integral is non-planar there are multiple ways to ``solve'' the momentum conservation constraint.  In this example they lead to the blue and purple parallelogram.  Therefore, we have \emph{two} dual points $x_B$ and $x_B'$ corresponding to one loop.}
    \label{fig:non-planar-dual}
\end{figure}

\section{Volume decomposition}
\label{sec:volume_decomposition}

The area of a triangle can be decomposed as
\begin{equation}
    \label{eq:area-decomposition}
    \langle 1 2 3\rangle = \langle A 2 3\rangle + \langle 1 A 3\rangle + \langle 1 2 A\rangle.
\end{equation}
Since these are oriented areas we do not need to impose that the point $x_A$ is inside the triangle.  We can check this decomposition by computing the boundaries
\begin{gather}
    \partial (1 2 3) = (2 3) - (1 3) + (1 2), \\
    \partial (A 2 3) = (2 3) - (A 3) + (A 2), \\
    \partial (1 A 3) = (A 3) - (1 3) + (1 A), \\
    \partial (1 2 A) = (2 A) - (1 A) + (1 2),
\end{gather}
so $\partial (1 2 3) = \partial ((A 2 3) + (1 A 3) + (1 2 A))$.

A problem one encounters often goes as follows.  Consider two triangles $(1 2 3)$ and $(1 2 4)$ which share a side $(1 2)$.  Suppose known all the sides of these two triangles.  What are the possible values of the length $(3 4)$?

One way to do this calculation involves angles; we can compute the angles $\angle(1 2 3)$ and $\angle(1 2 4)$, then $\angle(3 2 4) = \angle (1 2 3) \pm \angle(1 2 4)$.  However, this method fails if one of the sides is null.  Such situations do not arise in Euclidean geometry, but can arise in Lorentzian geometry.  Volumes are better behaved than angles in Lorentzian geometry so we will use them instead.

We impose that the volume of the tetrahedron $(1 2 3 4)$ vanishes
\begin{multline}
    -x_{12}^2 (x_{34}^2)^{2} -
    \Bigl((x_{12}^2)^{2} -
    x_{12}^2 x_{13}^2 -
    x_{12}^2 x_{14}^2 +
    x_{13}^2 x_{14}^2 -
    x_{12}^2 x_{23}^2 -
    x_{14}^2 x_{23}^2 -\\
    x_{12}^2 x_{24}^2 -
    x_{13}^2 x_{24}^2 +
    x_{23}^2 x_{24}^2\Bigr) x_{34}^2
   -x_{12}^2 x_{13}^2 x_{23}^2 +
    x_{12}^2 x_{14}^2 x_{23}^2 +
    x_{13}^2 x_{14}^2 x_{23}^2 -\\
    (x_{14}^2)^{2} x_{23}^2 -
    x_{14}^2 (x_{23}^2)^{2} +
    x_{12}^2 x_{13}^2 x_{24}^2 -
    (x_{13}^2)^{2} x_{24}^2 -
    x_{12}^2 x_{14}^2 x_{24}^2 +\\
    x_{13}^2 x_{14}^2 x_{24}^2 +
    x_{13}^2 x_{23}^2 x_{24}^2 +
    x_{14}^2 x_{23}^2 x_{24}^2 -
    x_{13}^2 (x_{24}^2)^{2}= 0.
\end{multline}
This is a second order equation in $x_{34}^2$ and therefore we obtain two solutions, as expected.

A more complicated situation is that of a pentagon obtained by gluing three triangles $(1 2 3)$, $(1 2 4)$ and $(2 3 5)$ with known side lengths.  Then we want to find the length $(4 5)$.  Obviously this length has \emph{four} possible values, which can be obtained by reflecting the point $4$ in the side $(1 2)$ and the point $5$ in the side $(2 3)$.

To find this distance we proceed in steps.  First, we find the possible values of the length $(3 4)$ as before.  Next, we apply the same result to the triangles $(2 3 4)$ and $(2 3 5)$.

A similar situation arises for the double box integral.  In that case we have tetrahedra instead and we can build other tetrahedra over two faces of an initial tetrahedron.  Then we want to impose a constraint on the distance between two vertices of these tetrahedra.

The volume of the tetrahedron can be decomposed as
\begin{equation}
    \label{eq:decomposition-dbox}
    \langle 1 2 3 4\rangle = \langle A 2 3 4\rangle + \langle 1 A 3 4\rangle + \langle 1 2 A 4\rangle + \langle 1 2 3 A\rangle.
\end{equation}
Since these are \emph{oriented} volumes we do not need to impose the condition that $x_A$ must be inside the tetrahedron for this formula to hold.

We can check that this is the correct decomposition by applying the boundary operator and checking that the ``internal'' boundaries cancel.  For example
\begin{gather}
    \partial (1 2 3 4) = (2 3 4) - (1 3 4) + (1 2 4) - (1 2 3), \\
    \partial (A 2 3 4) = (2 3 4) - (A 3 4) + (A 2 4) - (A 2 3), \\
    \partial (1 A 3 4) = (A 3 4) - (1 3 4) + (1 A 4) - (1 A 3), \\
    \partial (1 2 A 4) = (2 A 4) - (1 A 4) + (1 2 4) - (1 2 A), \\
    \partial (1 2 3 A) = (2 3 A) - (1 3 A) + (1 2 A) - (1 2 3).
\end{gather}
Then we have
\[
\partial ((A 2 3 4) + (1 A 3 4) + (1 2 A 4) + (1 2 3 A)) = \partial (1 2 3 4).
\]

\bibliographystyle{jhep}

\bibliography{Refs}

\end{document}